\documentclass[11pt,a4paper]{article}

\usepackage[margin=1in]{geometry}

\usepackage{theorem,latexsym,graphicx,amssymb}
\usepackage{amsmath,enumerate}
\usepackage{wrapfig}
\usepackage{subfigure}
\usepackage{float}
\usepackage{psfig}
\usepackage{epsfig}
\usepackage{xspace}
\usepackage{paralist}
\usepackage{times}
\usepackage[compact]{titlesec}
\usepackage[boxed,lined]{algorithm2e}
\usepackage{enumerate}
\usepackage{cases}
\usepackage{caption}
\usepackage{url}
\usepackage{tikz}
\usetikzlibrary{arrows}
\usetikzlibrary{topaths,calc}
\usepackage{booktabs}
\usepackage{mathtools}

\usepackage{color}


\usepackage{multicol}

\newenvironment{proof}{{\bf Proof:  }}{\hfill\rule{2mm}{2mm}}
\newenvironment{proofof}[1]{{\bf Proof of #1:  }}{\hfill\rule{2mm}{2mm}}

\numberwithin{figure}{section}
\numberwithin{equation}{section}

\newtheorem{theorem}{Theorem}[section]

\newtheorem{corollary}[theorem]{Corollary}
\newtheorem{lemma}[theorem]{Lemma}
\newtheorem{claim}[theorem]{Claim}

\newtheorem{example}[theorem]{Example}

\newcommand{\qed}{\hfill$\Box$}

\providecommand{\Appendix}{}
\renewcommand{\Appendix}[2][?]{%
        \refstepcounter{section}%
        \vspace{\parskip}%
        {\flushright\large\bfseries\appendixname\ \thesection: #1}%
        \vspace{\baselineskip}%
}

\renewcommand{\appendix}{%
        \newpage
        \renewcommand{\section}{\secdef\Appendix\Appendix}%
        \renewcommand{\thesection}{\Alph{section}}%
        \setcounter{section}{0}%
}

\newcommand{\eps}{\epsilon}

\newcommand{\vp}{\varphi}

\newcommand{\R}{\mathbb{R}}

\newcommand{\spn}{\operatorname{span}}

\newcommand{\T}{\ensuremath{\mathsf{T}}\xspace}

\newcommand{\W}{\ensuremath{\mathsf{W}}\xspace}
\newcommand{\Wm}{\ensuremath{\mathsf{W}^{-1}}\xspace}

\newcommand{\Wh}{\ensuremath{\mathsf{W}^{\frac{1}{2}}}\xspace}

\newcommand{\Wmh}{\ensuremath{\mathsf{W}^{-\frac{1}{2}}}\xspace}

\newcommand{\Lo}{\ensuremath{\mathsf{L}\xspace}}
\newcommand{\Lc}{\ensuremath{\mathcal{L}}\xspace}

\newcommand{\D}{\ensuremath{\mathsf{D}\xspace}}
\newcommand{\Dc}{\ensuremath{\mathcal{D}}\xspace}

\newcommand{\M}{\ensuremath{\mathsf{M}\xspace}}

\newcommand{\I}{\ensuremath{\mathsf{I}\xspace}}

\newcommand{\ray}{\ensuremath{\mathsf{R}}\xspace}
\newcommand{\rayc}{\ensuremath{\mathcal{R}}\xspace}

\newcommand{\Ito}{It\={o}\xspace}

\newcommand{\w}[1]{\ensuremath{w_{#1}}}
\newcommand{\f}[1]{\ensuremath{f_{#1}}}
\newcommand{\g}[1]{\ensuremath{g_{#1}}}

\newcommand{\ra}{\rightarrow}

\newcounter{note}[section]

\newcommand{\initOneLiners}{%
    \setlength{\itemsep}{0pt}
    \setlength{\parsep }{0pt}
    \setlength{\topsep }{0pt}
}

\newcommand{\ignore}[1]{}

\newcommand*\samethanks[1][\value{footnote}]{\footnotemark[#1]}

\newcommand{\shortv}[1]{}

\title{Spectral Properties of Laplacian and Stochastic Diffusion Process for Edge Expansion in Hypergraphs}

\author{T-H. Hubert Chan\thanks{Department of Computer Science, the University of Hong Kong. {\texttt{\{hubert,zhtang,czzhang\}@cs.hku.hk}}} \and Zhihao Gavin Tang\samethanks \and Chenzi Zhang\samethanks}
\date{}

\begin{document}

\begin{titlepage}

\maketitle

\begin{abstract}
There has been recent work [Louis STOC 2015] to analyze the spectral properties of hypergraphs with respect to edge expansion.  In particular, a diffusion process is defined on a hypergraph such that within each hyperedge, measure flows from nodes having maximum weighted measure to those having minimum.  The diffusion process determines a Laplacian, whose spectral properties are related to the edge expansion properties of the hypergraph.

It is suggested that in the above diffusion process, within each hyperedge, measure should flow uniformly in the complete bipartite graph from nodes with maximum weighted measure to those with minimum.  However, we discover that this method has some technical issues.  First, the diffusion process would not be well-defined.  Second, the resulting Laplacian would not have the claimed spectral properties.

In this paper, we show that the measure flow between the above two sets of nodes must be coordinated carefully over different hyperedges in order for the diffusion process to be well-defined, from which a Laplacian can be uniquely determined.  Since the Laplacian is non-linear, we have to exploit other properties of the diffusion process to recover a spectral property concerning the ``second eigenvalue'' of the resulting Laplacian.  Moreover, we show that higher order spectral properties cannot hold in general using the current framework.

Inspired from applications in finance, we consider a stochastic diffusion process, in which each node also experiences Brownian noise from outside the system.  We show a relationship between the second eigenvalue and the convergence behavior of the process.  In particular, for the special case when the Brownian noise at each node has zero variance, the process reduces to the (deterministic) diffusion process within a closed system, and we can recover an upper bound on the mixing time in terms of the second eigenvalue.
\end{abstract}

\thispagestyle{empty}
\end{titlepage}

\section{Introduction}
\label{sec:intro}

Recently, spectral analysis of edge expansion
has been extended
from normal graphs to hypergraphs in a STOC 2015 paper~\cite{Louis15stoc}.
Let $H = (V,E)$ be a hypergraph on $n = |V|$ nodes
with non-negative edge weights $w : E \rightarrow \R_+$.
We say $H$ is a $k$-graph if every edge contains exactly $k$ nodes.  (Hence,
a normal graph is a 2-graph.)
Each node $v \in V$ has weight $w_v := \sum_{e \in E: v \in e} w_e$. A subset $S$ of nodes has
weight $w(S) := \sum_{v \in S} w_v$, and the edges it cuts is $\partial S := \{e \in E : $ $e$
intersects both $S$ and $V \setminus S \}$.
The \emph{edge expansion} of $S \subset V$
is defined as $\phi(S) := \frac{w(\partial S)}{w(S)}$.

In classical spectral graph theory, the edge expansion
is related to the \emph{discrepancy ratio}, which is defined as $\D_w(f)=\frac{\sum_{e\in E} \w{e} \max_{u,v\in e}{(\f{u}-\f{v})^2}}{\sum_{u \in V} \w{u} \f{u}^2}$,
for each non-zero vector $f \in \R^V$.
Observe that if $f$ is the indicator vector for a subset $S \subset V$,
then $\D_w(f) = \phi(S)$.

One often considers the transformation into
the \emph{normalized space} given by 
$x := \Wh f$, where $\W$ is the diagonal matrix
whose $(v,v)$-th entry is $w_v$ for $v \in V$.
The normalized discrepancy ratio is
$\Dc(x) := \D_w(\Wmh x) = \D_w(f)$.
A well-known result~\cite{chung1997spectral} is that the \emph{normalized
Laplacian} for a 2-graph can be defined as
$\Lc := \I - \Wmh A \Wmh$ (where $\I$ is the
identity matrix, and $A$ is the symmetric matrix
giving the edge weights)
such that its \emph{Rayleigh quotient}
$\rayc(x) := \frac{\langle x, \Lc x \rangle}{\langle x, x \rangle}$ coincides with $\Dc(x)$.
Since the matrix $\Lc$ is symmetric for $2$-graphs,
its eigenvalues are given by $\gamma_i$'s as follows.

\noindent \emph{Procedural Minimizers.}
Define $x_1 := \Wh \vec{1}$, where $\vec{1} \in \R^V$ is the all-ones vector; $\gamma_1 := \Dc(x_1) = 0$.
Suppose $\{(x_i, \gamma_i)\}_{i \in [k-1]}$ have
been constructed\footnote{In this paper, for a positive
integer $s$, we denote $[s] := \{1,2,\ldots, s\}$.}.  Define $\gamma_k := \min \{ \Dc(x) : \vec{0} \neq x \perp \{x_i:
i \in [k-1]\}\}$, and $x_k$ to be any such minimizer that attains $\gamma_k = \Dc(x_k)$.

One of the main results in~\cite{Louis15stoc}
is an attempt to define a Laplacian $\Lc$ for a hypergraph
and relate its spectral properties with the $\gamma_i$'s produced
by the corresponding procedural minimizers.
However, we have discovered some technical issues
with their construction and proofs, which we outline
as follows.

\noindent \emph{Defining Operator via Diffusion Process.}
The nodes have measure that is indicated by a vector $\vp \in \R^V$
in the \emph{measure space}.
As in the case for 2-graphs, in the equilibrium distribution,
the measure at each node is proportional to its weight.  Hence,
we consider the weighted vector $f := \Wm \vp$,
and for each edge $e \in E$, its discrepancy $\Delta_e := \max_{u,v \in e} (f_u - f_v)$ indicates how far the measure for nodes in $e$
is from the equilibrium.  We can imagine that nodes in each $e \in E$ have formed
a pact such that if the discrepancy~$\Delta_e$ is non-zero,
then some measure should flow from the nodes~$S_e$ having maximum
$f$ values in $e$ to the nodes~$I_e$ having minimum $f$ values in $e$.
Moreover, the total rate of measure flow due to~$e$ is $c_e := w_e \cdot \Delta_e$.
One can view this as distributing the weight $w_e$ of edge $e$ among
pairs in the bipartite graph $S_e \times I_e$ to produce
a symmetric matrix $A_f$, whose
$(u,v)$-th entry is the weight collected by the pair $\{u,v\}$ from
all edges $e \in E$.
As we shall see in Lemma~\ref{lemma:ray_disc}, this ensures that
the Rayleigh quotient of the resulting Laplacian 
$\Lc x := (\I -  \Wmh A_f \Wmh)x$, where $x = \Wh f$,
coincides with the
discrepancy ratio.

\noindent \textbf{Issues with Hyperedge Weight Distribution in~\cite{Louis15stoc}.}
It is explained in~\cite{Louis15stoc} that if $|S_e| \times |I_e| > 1$,
then the weight $w_e$ cannot be given to just one pair $(u,v) \in S_e \times I_e$.
Otherwise, for the case $|S_e| \geq 2$, after infinitesimal time, 
among nodes in $S_e$, only $\vp_u$ (and $f_u$) will decrease due to $e$.
This means $u$ will no longer be in $S_e$ after infinitesimal time,
and we have to pick another node in $S_e$ immediately. However,
we can run into the same problem again if we try to pick another node
from $S_e$, and the diffusion process cannot continue.
There are a couple of suggested approaches.

\begin{compactitem}
\item[1.] In the conference version~\cite{Louis15stoc},
a pair $(u,v)$ is picked randomly in $S_e \times I_e$ to receive
the weight~$w_e$.  However, this produces an operator that returns
a random vector.  If one considers its expectation, then
this is equivalent to distributing the weight $w_e$ evenly among
the pairs in $S_e \times I_e$; this is the approach
explicitly stated in the arXiv version~\cite{louis2014hypergraph}.
We denote by $\overline{\Lc}$ the normalized Laplacian achieved by this method.

\item[2.] Both versions mentioned that a ``suitably weighted'' bipartite
graph on $S_e \times I_e$ can be added.  However, it is ambiguous whether
this refers to the bipartite graph with uniform edge weights as above or some other weights.  In any case, in~\cite[Theorem 4.6]{louis2014hypergraph},
the proof assumes that even when $\vp$ and $f$ change continuously,
there can only be a finite number of resulting $A_f$'s.
\end{compactitem}

We have discovered the following issues using the above weight distribution approaches.

\begin{compactitem}
\item[1.] \textbf{Diffusion process is not well-defined.}
We illustrate an issue if
the weight $w_e$ is distributed evenly among pairs
in $S_e \times I_e$.
In Example~\ref{eg:Louis},
there is an edge $e_5 = \{a, b, c\}$
such that the node in $I_{e_5} = \{c\}$
receives measure from the nodes in $S_{e_5} = \{a,b\}$.
However, node $b$ also gives some measure to node 
$d$ because of the edge $e_2 = \{b, d\}$.
In the example, all nodes have the same weight.
Now, if $w_{e_5}$ is distributed evenly among $\{a,c\}$ and $\{b, c\}$,
then the measure of $a$ decreases more slowly than that of $b$
because $b$ loses extra measure due to $e_2$.
Hence, after infinitesimal time, $b$ will no longer be in $S_{e_5}$.
This means that the measure of $b$ should not have been decreased at all due to $e_5$,
contradicting the choice of distributing $w_{e_5}$ evenly.

\item[2.] \textbf{Spectral properties of $\overline{\Lc}$ are not as claimed.}
Using the same Example~\ref{eg:Louis},
we can show that for all non-zero minimizers $x_2$
attaining $\gamma_2 := \min_{\vec{0} \neq x \perp x_1} \Dc(x)$,
where $x_1 := \Wh \vec{1}$, the vector
$x_2$ is not an eigenvector of the operator $\overline{\Lc}$
or even $\Pi_{x_1^\perp} \overline{\Lc}$,
where $\Pi_S$ is the projection operator into the subspace spanned by $S$.
This is contrary to~\cite[Proposition 2.7]{Louis15stoc}
and the proof of~\cite[Theorem~4.7]{louis2014hypergraph}.

\item[3.] \textbf{Matrix $A_f$ changes continuously.} As we see later,
using the framework of distributing weight $w_e$
among pairs $S_e \times I_e$, one needs to consider continuous change
in $A_f$.  Hence, one cannot conclude that there is some non-empty time
interval such that $A_f$ stays the same.  However, this conclusion
is needed in the proof of~\cite[Theorem~4.6]{louis2014hypergraph}.
Incidentally, the lemma~\cite[Lemma~B.2]{louis2014hypergraph} used to prove this conclusion is also inaccurate,
because one can consider a function on the interval $[0,1]$ that maps
rationals to 1 and irrationals to 2, and observe that any non-empty interval
contains both rationals and irrationals.
\end{compactitem}

\subsection{Our Contributions and Results}

We also consider a diffusion process
by using the hyperedge weight distribution framework~\cite{Louis15stoc}
described above,
and show that indeed there is a suitable way to distribute
the weight of every hyperedge~$e$ among pairs in the bipartite
graph $S_e \times I_e$.  We resolve the above issues
by showing that the diffusion process is well-defined and
a unique normalized Laplacian can be determined, whose
spectral properties can be related to the discrepancy ratio.
Our main result is as follows.

\begin{theorem}[Diffusion Process and Laplacian]
\label{th:main1}
Using the hyperedge weight distribution framework~\cite{Louis15stoc}
to consider a diffusion process, a unique normalized Laplacian $\Lc$ (that
is not necessarily linear)
can be defined on the normalized space such that the
following holds.
\begin{compactitem}
\item[1.] For all $\vec{0} \neq x \in \R^V$, the Rayleigh quotient $\frac{\langle x, \Lc x \rangle}{\langle x, x \rangle}$ coincides with the discrepancy ratio $\Dc(x)$.

\item[2.] There is an operator $\Lo := \Wh \Lc \Wmh$ on the measure space
such that the diffusion process
can be described by the differential equation $\frac{d \vp}{d t} = - \Lo \vp$.

\item[3.] Any procedural minimizer $x_2$ attaining $\gamma_2 := \min_{\vec{0} \neq
x \perp \Wh \vec{1}} \Dc(x)$ satisfies $\Lc x_2 = \gamma_2 x_2$.

\item[4.] For some hypergraph,
for all procedural minimizers $\{x_1, x_2\}$,
any procedural minimizer $x_3$ attaining
$\gamma_3 := \min_{\vec{0} \neq
x \perp \{x_1, x_2\}} \Dc(x)$ is not
an eigenvector of $\Pi_{\{x_1, x_2\}^\perp} \Lc$.
\end{compactitem}
\end{theorem}

The first three statements are proved in Lemmas~\ref{lemma:ray_disc}, 
\ref{lemma:define_lap} and Theorem~\ref{th:hyper_lap}.
The fourth statement is proved by Example~\ref{eg:gamma3},
and suggests that the current approaches cannot be generalized
to consider higher order eigenvalues of the Laplacian $\Lc$.
In other words, one cannot hope to achieve~\cite[Proposition 2.7]{Louis15stoc}
using the current framework, as our Laplacian $\Lc$ is uniquely determined.  On the other hand,
we remark that in our third statement,
\emph{any} minimizer $x_2$ attaining $\gamma_2$
is an eigenvector of $\Lc$, whereas in~\cite[Theorem~4.6]{louis2014hypergraph}, it is only claimed
that there is \emph{some} minimizer $x_2$ that is an eigenvector.

\noindent \textbf{Our Techniques.}  As seen in Example~\ref{eg:Louis}, the trickiest part
in defining the diffusion process is when
nodes in the same hyperedge $e$ have the same
$f$ value.  We consider an equivalence relation
on $V$, where nodes in the same equivalence class $U$
have the same $f$ value.  The crucial part in the
analysis is to decide after infinitesimal time,
which nodes in $U$ will remain in the same equivalence class,
and which ones will go separate ways.
Suppose at the moment, the subset $X \subseteq U$ will have the maximum rate of change
in their $f$ values.
Then, after infinitesimal time, the nodes in $X$
will have larger $f$ value than
the rest of $U$.  Hence, the nodes in $X$ can only receive measure
from the set $I_X$ of hyperedges $e$ such that $I_e \subseteq X$,
but they will lose measure due to
the set $S_X$ of hyperedges $e$ such that $S_e \cap X \neq \emptyset$.
Hence, 
for $X \subseteq U$,
we consider a density function $\delta(X) := \frac{c(I_X) - c(S_X)}{w(X)}$.
In Lemma~\ref{lemma:define_lap},
it turns out that the maximal set $T \subseteq U$ having maximum $\delta(T)$ is well-defined and unique,
and a careful argument shows that all nodes in $T$ indeed have their rate of change of $f$ value
being $\delta(T)$.

Having defined the normalized Laplacian $\Lc$,
we extract a structural property of the diffusion process
to derive an exact expression for the rate of change
of the Rayleigh quotient $\rayc(x)$, which is non-positive,
and attains 0 \emph{iff} $\Lc x \in \spn(x)$.
In Theorem~\ref{th:hyper_lap}, we argue that any minimizer $x_2$ attaining $\gamma_2$ must have zero rate of change
of Rayleigh quotient at the moment, thereby showing
that $x_2$ is an eigenvector of $\Lc$.

\noindent \textbf{Stochastic Diffusion Process.}
The diffusion process defined so far is deterministic,
and no measure enters or leaves the system.  We believe that
it will be of independent interest to consider
the case when each node can experience independent noise
from outside the system, for instance, in risk management 
applications~\cite{merton1969lifetime,merton1971optimum}.  Since
the diffusion process is continuous in nature,
we consider Brownian noise.

For some $\eta \geq 0$, we assume that
the noise experienced by each node $u$
follows the Brownian motion whose rate of variance
is $\eta w_u$.  Then, the measure $\Phi_t \in \R^V$
of the system is an \Ito process
defined by the stochastic
differential equation $d \Phi_t = - \Lo \Phi_t \, dt + \sqrt{\eta} \cdot \Wh \, d B_t$.
For $\eta = 0$, this reduces to the (deterministic) diffusion process in a closed system, and we can recover the upper bound on the \emph{mixing time} in terms
of $\gamma_2$ in Corollary~\ref{cor:mixing},
which is also claimed in~\cite{Louis15stoc}.

The interesting question is whether such a relationship
between $\gamma_2$ and the system's convergence behavior
can be extended to the stochastic diffusion process.  Given a measure vector $\vp \in \R^V$,
we denote by $\vp^*$ the equilibrium vector
obtained by distributing the total measure $\sum_{v \in V} \vp_v$ among nodes proportional to their weights.
Hence, $\|\vp - \vp^*\|_1$ gives a metric of how far
the measure vector $\vp$ is from the equilibrium.
We show the following theorem in Section~\ref{sec:stochastic} that relates $\gamma_2$ with the convergence behavior of the process.

\begin{theorem}[Convergence and Laplacian]
Suppose the stochastic diffusion process is given
by the equation $d \Phi_t = - \Lo \Phi_t \, dt + \sqrt{\eta} \cdot \Wh \, d B_t$,
with some initial measure $\Phi_0$.  Then,
as $t$ tends to infinity,
$\|\Phi_t - \Phi_t^*\|_1^2$ is stochastically dominated by $\frac{\eta \cdot w(V)}{2 \gamma_2} \cdot \chi^2(n)$,
where $\chi^2(n)$ is the
chi-squared distribution with $n$ degrees of freedom.
In particular, $\lim_{t \ra \infty} E[ \|\Phi_t - \Phi_t^*\|_1]
\leq \sqrt{\frac{\eta n \cdot w(V)}{2 \gamma_2}}$.
\end{theorem}

\emph{Organization of Paper.} The definition of the
diffusion process and the spectral properties of the normalized Laplacian are given in Section~\ref{sec:laplacian}.
For better readability, the examples are given in Appendix~\ref{sec:example}.

\subsection{Related Work}

Naturally, the most related work is the recent
STOC paper by Louis~\cite{Louis15stoc},
which includes a comprehensive review of the related
literature.  We only give a brief summary here.

\emph{Spectral analysis for 2-graphs.}
The spectral properties of 2-graphs in relation to cuts
have been studied extensively (see the background surveys~\cite{chung1997spectral,DBLP:journals/fttcs/MontenegroT05}).
In particular, the Cheeger inequality for bi-partitions~\cite{cheeger1970lower, alon1986eigenvalues, alon1985lambda} has been recently
generalized to higher order Cheeger inequalities
for multi-partitions~\cite{DBLP:conf/approx/LouisRTV11,DBLP:conf/stoc/LouisRTV12,
DBLP:conf/stoc/KwokLLGT13,DBLP:journals/jacm/LeeGT14,DBLP:conf/soda/LouisM14}.

\emph{Algorithm and complexity results relating to cut problems.}  Besides analyzing cut properties using spectral techniques, there have been numerous works
on algorithms~\cite{DBLP:journals/jacm/AroraRV09,DBLP:conf/stoc/RaghavendraST10,DBLP:journals/siamcomp/BansalFKMNNS14} and complexity~\cite{DBLP:conf/stoc/RaghavendraS10} results related to cuts and small set expansion.

\emph{Cut Problems on Hypergraphs.}  Various cut problems
have been generalized to hypergraphs~\cite{lawler1973cutsets,DBLP:conf/esa/EneN14,louis2014approximation}, which have applications such as 
VLSI circuit design~\cite{DBLP:journals/tvlsi/KarypisAKS99}.

\emph{Laplacians for Hypergraphs.}  There are several attempts
to define Laplacians for hypergraphs, but their properties
are not applicable to the discrepancy ratio $\D_w$ defined
to capture hypergraph set expansion.  For instance,
Chung~\cite{chung1993laplacian} defined a Laplacian with very different
structural properties relating to homologies.  Friedman and Wigderson~\cite{DBLP:journals/combinatorica/FriedmanW95} considered spectral properties
of tensors for hypergraphs, but it is shown~\cite{louis2014hypergraph}
that they are not related to hypergraph expansion.
Rodriguez~\cite{DBLP:journals/appml/Rodriguez09} considered replacing
a hyperedge by a clique, but this approach measures how even each hyperedge is cut,
while hyperedge expansion just concerns whether a hyperedge is cut or not.

\emph{Min-max operator in process definition.}  In the definition
of our diffusion process, observe that the rule to determine
which nodes are currently actively participating depends on 
the minimum or the maximum of some attributes.  This characteristic
has also appeared in the context of bargaining networks~\cite{DBLP:conf/wine/CelisDP10} and another notion of Laplacian~\cite{peres2009tug}.





\noindent \emph{Stochastic diffusion process.} This is a well-studied subject,
with applications in physics and finance.  The reader can refer to 
standard textbooks~\cite{gardiner1985handbook,oksendalSDE} for the relevant background.

%

\noindent \emph{Other results in~\cite{Louis15stoc}.}
Although there are issues with the diffusion process
and the spectral properties of the Laplacian,
the algorithmic results concerning the procedural minimizers,
sets with small edge expansion and sparsest cut with general demands
can still hold.
Moreover, the results on the higher order Cheeger inequalities
can also be rephrased
using only properties of the discrepancy ratios $\D_w$ or $\Dc$ (but without
connection to the diffusion process or the spectral properties of the Laplacian) via the following parameters.

\noindent \textbf{Orthogonal Minimaximizers.}  
Define $\xi_k := \min_{x_1, \ldots x_k} \max_{i \in [k]} \Dc(x_i)$
and $\zeta_k := \min_{x_1, \ldots x_k}  \max \{\Dc(x) : x \in \spn\{x_1, \ldots x_k\}\}$, where the minimum is over $k$ non-zero mutually orthogonal vectors $x_1, \ldots, x_k$ in
the normalized space.  (All involved minimum and maximum 
can be attained because $\Dc$ is continuous and all vectors could be
chosen from the surface of a unit ball, which is compact.)

For 2-graphs, the three parameters $\xi_k = \gamma_k = \zeta_k$
coincide with the eigenvalues of the normalized Laplacian~$\Lc$.  
Indeed,
most proofs in the literature concerning expansion and Cheeger inequalities  (e.g.,~\cite{DBLP:journals/jacm/LeeGT14,DBLP:conf/stoc/KwokLLGT13,Louis15stoc})
just need to use the underlying properties of $\gamma_k$, $\xi_k$ and $\zeta_k$ with
respect to the discrepancy ratio,
without explicitly using the spectral properties of the Laplacian.
For hypergraphs, we show in Example~\ref{ex:gamma_nonunique} that
the sequence $\{\gamma_k\}$ might not even be unique.
However, the three parameters
can be related to one another in the following lemma,
whose proof is in Appendix~\ref{sec:min}.

\begin{lemma}[Comparing Discrepancy Minimizers]
\label{lemma:min}
Suppose $\{\gamma_k\}$ is some sequence produced by the procedural minimizers.
For each $k \geq 1$, $\xi_k \leq \gamma_k \leq \zeta_k \leq k \xi_k$.  
In particular, $\gamma_2 = \zeta_2$, but
it is possible that $\xi_2 < \gamma_2$.
\end{lemma}

\section{Preliminaries}
\label{sec:prelim}

Without loss of generality, we assume that all nodes
have positive weights, since any node with zero weight can be removed.
We use $\R^V$ to denote the set of column vectors.
Given $f \in \R^V$,
we use $f_u$ or $f(u)$ (if we need to use the subscript to distinguish between different vectors) to indicate the coordinate
corresponding to $u \in V$.
We use $A^\T$ to denote the transpose of a matrix $A$.
We use three isomorphic spaces described as follows.

\noindent \textbf{Weighted Space.} This is the space
associated with the discrepancy ratio $\D_w$ to consider
edge expansion.
For $f,g \in \R^V$, the inner product
is defined as $\langle f, g \rangle_w := f^\T \W g$,
and the associated norm is $\| f \|_w := \sqrt{\langle f, f \rangle_w}$.  We use $f \perp_w g$ to mean $\langle f, g \rangle_w = 0$.

\noindent \textbf{Normalized Space.} 
Given $f \in \R^V$ in the weighted space,
the corresponding vector in the normalized space is $x := \Wh f$.
In the normalized space, the usual $\ell_2$ inner product and norm are used.
Observe that if $x$ and $y$ are the corresponding
normalized vectors for $f$ and $g$ in the weighted space,
then $\langle x, y \rangle = \langle f, g \rangle_w$.

\noindent \textbf{Measure Space.}  This is the space
associated with the diffusion process.
Given $f$ in the weighted space,
the corresponding vector in the measure space
is given by $\vp := \W f$.  
Observe that a vector in the measure space can have negative coordinates.
We do not consider inner product explicitly in this space, and 
so there is no special notation for it.
However, we use the $\ell_1$-norm, which is not
induced by an inner product.
For vectors $\vp_i = \Wh x_i$,
an application of the Cauchy-Schwarz inequality implies that
$\|\vp_1 - \vp_2\|_1 \leq \sqrt{w(V)} \cdot \|x_1 - x_2\|_2$.

\noindent \textbf{Notation.}
We use the Roman letter $f$ for vectors in
the weighted space, $x$ for vectors in the
normalized space, and Greek letter $\vp$
for vectors in the measure space.
Observe that 
an operator defined on one space
induces operators on
the other two spaces.
For instance, if $\Lo$ is an operator defined on the measure space,
then $\Lo_w := \Wm \Lo \W$  is the corresponding
operator on the weighted space 
and $\Lc := \Wmh \Lo \Wh$ is the one on the normalized space.
Moreover, all three operators have the same eigenvalues.
The Rayleigh quotients are defined
as $\ray_w(f) := \frac{\langle f, \Lo_w f \rangle_w}{\langle f, f \rangle_w}$
and $\rayc(x) := \frac{\langle x, \Lc x \rangle}{\langle x, x \rangle}$. For $\Wh f = x$, we have $\ray_w(f) = \rayc(x)$.

Given a set $S$ of vectors in the normalized
space, $\Pi_S$ is the orthogonal projection operator
onto the subspace spanned by $S$.  The orthogonal
projection operator $\Pi^w_S$ can also be defined for
the weighted space.

\section{Defining Diffusion Process and Laplacian for Hypergraphs}
\label{sec:laplacian}

A classical result in spectral graph theory
is that for a $2$-graph whose edge weights
are given by the adjacency matrix $A$,
the parameter $\gamma_2 := \min_{\vec{0} \neq x \perp  \Wh \vec{1}} \Dc(x)$ is an eigenvalue of the normalized
Laplacian \mbox{$\Lc := \I - \Wmh A \Wmh$}, 
where a corresponding minimizer $x_2$ is an eigenvector
of $\Lc$.
Observe that $\gamma_2$ is also an eigenvector
on the operator $\Lo_w := \I - \Wm A$ induced on the weighted space.
However, 
in the literature, the (weighted) Laplacian is defined
as $\W - A$, which is $\W \Lo_w$ in our notation.  Hence,
to avoid confusion, we only consider the normalized Laplacian in this paper.

In this section, we generalize the result to hypergraphs.
Observe that any result for the normalized space has an equivalent counterpart in the weighted space, and vice versa.

\begin{theorem}[Eigenvalue of Hypergraph Laplacian]
\label{th:hyper_lap}
For a hypergraph with edge weights $w$,
there exists a normalized Laplacian $\Lc$ such that the
normalized discrepancy ratio $\Dc(x)$ coincides
with the corresponding Rayleigh quotient $\rayc(x)$.
Moreover, 
the parameter $\gamma_2 := \min_{\vec{0} \neq x \perp  \Wh \vec{1}} \Dc(x)$ is an eigenvalue of $\Lc$,
where any minimizer $x_2$ is a corresponding eigenvector.
\end{theorem}

However, contrary to what is claimed
in~\cite{Louis15stoc}, we show in Example~\ref{eg:gamma3}
that the above result for our Laplacian does not hold for $\gamma_3$.

\noindent \textbf{Intuition from Random Walk and Diffusion Process.}  Given a $2$-graph whose edge weights
$w$ are given by the (symmetric) matrix $A$,
we first illustrate the relationship between the Laplacian
and a diffusion process in an underlying measure space,
in order to gain insights on how to define the Laplacian
for hypergraphs.

Suppose $\vp \in \R^V$ is some measure on the nodes,
which, for instance, can represent a probability distribution
on the nodes.  A random walk on the graph
can be characterized by the transition matrix $\M := A \Wm$.
Observe that each column of $\M$ sums to 1,
because we apply $\M$ to the column vector $\vp$
to get the distribution $\M \vp$ after one step of the random walk.

We wish to define a continuous diffusion process.
Observe that, at this moment, the measure vector $\vp$ is moving
in the direction of $\M \vp - \vp = (\M - \I) \vp$.
Therefore, if we define an operator $\Lo := \I - \M$
on the measure space, we have
the differential equation $\frac{d \vp}{d t} = - \Lo \vp$.

To be mathematically precise, we are considering how $\vp$
will move in the future.  Hence, unless otherwise stated,
all derivatives considered are actually right-hand-derivatives
$\frac{d \vp(t)}{dt} := \lim_{\Delta t \ra 0^+} \frac{\vp(t + \Delta t) - \vp(t)}{\Delta t}$.

Using the transformation into
the weighted space $f = \Wm \vp$
and the normalized space $x = \Wmh \vp$,
we can define the corresponding operators
$\Lo_w := \Wm \Lo \W = \I - \Wm A$
and $\Lc := \Wmh \Lo \Wh = \I - \Wmh A \Wmh$,
which is exactly the normalized Laplacian for $2$-graphs.

\noindent \textbf{Generalizing the Diffusion Rule from 2-Graphs to
Hypergraphs.}  We
consider more carefully the
rate of change for the measure at a certain node $u$:
$\frac{d \vp_u}{d t} = \sum_{v : \{u,v\} \in E} w_{uv} (f_v - f_u)$, where $f = \Wm \vp$ is the weighted measure.
Observe that for a stationary distribution of the random walk,
the measure at a node $u$ should be proportional
to its (weighted) degree $w_u$.
Hence, given an edge $e = \{u,v\}$, by comparing the values $f_u$ and $f_v$, measure should move from the node with higher $f$ value
to the node with smaller $f$ value, at the rate given by $c_e := w_e \cdot |f_u - f_v|$.

To generalize this to a hypergraph $H = (V,E)$,
for $e \in E$ and measure $\vp$ (corresponding to $f = \Wm \vp$),
we define $I_e(f) \subseteq e$ as the nodes $u$
in $e$ whose $f_u = \frac{\vp_u}{w_u}$ are minimum, $S_e(f) \subseteq e$ as those whose corresponding values
are maximum, and $\Delta_e(f) := \max_{u,v \in E} (f_u - f_v)$
as the discrepancy within edge $e$.  Then,
the diffusion process obeys the following rules.

\begin{compactitem}
\item[\textsf(R1)] When the measure distribution is at state $\vp$ (where $f = \Wm \vp$), there can be a positive rate of measure flow
from $u$ to $v$ due to edge $e \in E$ only if  $u \in S_e(f)$ and $v \in I_e(f)$.
\item[\textsf(R2)] For every edge $e \in E$,
the total rate of measure flow \textbf{due to $e$} from nodes
in $S_e(f)$ to $I_e(f)$ is $c_e := w_e \cdot \Delta_e(f)$.
In other words,
the weight $w_e$ is distributed among $(u,v) \in S_e(f) \times I_e(f)$ such that 
for each such $(u,v)$,
there exists $a^e_{uv} = a^e_{uv}(f)$
such that $\sum_{(u,v) \in S_e \times I_e} a^e_{uv} = w_e$,
and
the rate of flow from $u$ to $v$ (due to $e$) is $a^e_{uv} \cdot \Delta_e$. (For ease of notation, we write $a^e_{uv} = a^e_{vu}$.)
Observe that if $I_e = S_e$, then $\Delta_e = 0$ and it does not matter how
the weight $w_e$ is distributed.
\end{compactitem}

Observe that the distribution of hyperedge weights will induce a symmetric matrix $A_f$
such that for $u \neq v$, $A_f(u,v) = a_{uv} := \sum_{e \in E} a^e_{uv} (f)$, and the diagonal entries are chosen such that entries in  the row corresponding to node $u$ sum to $w_u$.
Then, the operator $\Lo(\vp) := (\I - A_f \Wm) \vp$
is defined on the measure space to obtain
the differential equation $\frac{d \vp}{d t} = - \Lo \vp$.
As in the case for $2$-graph, we show in Lemma~\ref{lemma:ray_disc} that
the corresponding operator $\Lo_w$ on the weighted space
and the normalized Laplacian $\Lc$ are induced such that
$\D_w(f) = \ray_w(f)$ and $\Dc(x) = \rayc(x)$,
which hold no matter how the weight $w_e$ of hyperedge $e$
is distributed among edges in $S_e(f) \times I_e(f)$.

\begin{lemma}[Rayleigh Quotient Coincides with Discrepancy Ratio]
\label{lemma:ray_disc}
Suppose $\Lo_w$ on the weighted space is defined such that
rules \textsf{(R1)} and \textsf{(R2)} are obeyed.
Then, 
the Rayleigh quotient associated with $\Lo_w$ satisfies
that for any $f$ in the weighted space,
$\ray_w(f) = \D_w(f)$.
By considering the isomorphic normalized space,
we have for each $x$, $\rayc(x) = \Dc(x)$.
\end{lemma}

\begin{proof}
It suffices to show that $\langle f, \Lo_w f \rangle_w =
\sum_{e \in E} w_e \max_{u,v \in e} (f_u - f_v)^2$.

Recall that $\vp = \W f$,
and $\Lo_w = \I - \Wm A_f$,
where $A_f$ is chosen as above
to satisfy rules~\textsf{(R1)} and~\textsf{(R2)}.

Hence,
it follows that

$\langle f, \Lo_w f \rangle_w = f^\T (\W - A_f) f
= \sum_{uv \in {V \choose 2}} a_{uv} (f_u - f_v)^2$

$= \sum_{uv \in {V \choose 2}} \sum_{e \in E: \{uv, vu\} \cap S_e \times I_e \neq \emptyset} a^e_{uv} (f_u - f_v)^2
= \sum_{e \in E} w_e \max_{u,v \in e} (f_u - f_v)^2,
$
as required.
\end{proof}

\subsection{Defining Diffusion Process to Construct Laplacian}
\label{sec:disp}

Recall that $\vp \in \R^V$ is the measure vector, where each coordinate
contains the ``measure'' being dispersed.  Observe that
we consider a closed system here, and hence $\langle \vec{1}, \vp \rangle$ remains invariant.
To facilitate the analysis, we also consider the weighted measure $f := \Wm \vp$.

Our goal is to define a diffusion process
that obeys rules~\textsf{(R1)} and~\textsf{(R2)}.
Then, the operator on the measure space
is given by $\Lo \vp := - \frac{d \vp}{d t}$.
By observing that the weighted space is achieved
by the transformation $f = \Wm \vp$,
the operator on the weighted space is
given by $\Lo_w f := - \frac{d f}{d t}$.
The Laplacian is induced $\Lc := \Wmh \Lo \Wh$
on the normalized space $x = \Wh f$.

Suppose we have the measure vector $\vp \in \R^V$
and the corresponding weighted vector $f = \Wm \vp$.
Observe that even though we call $\vp$ a measure vector,
$\vp$ can still have negative coordinates.
We shall construct a vector $r \in \R^V$ that is
supposed to be $\frac{df}{dt}$.
For $u \in V$ and $e \in E$,
let $\rho_u(e)$ be the rate of change of
the measure $\vp_u$ due to edge $e$.  Then, $\rho_u :=
\sum_{e \in E} \rho_u(e)$ gives the rate of change of $\vp_u$.

We show that $r$ and $\rho$ must satisfy certain
constraints because of
rules~\textsf{(R1)} and~\textsf{(R2)}.
Then, it suffices to show that there exists 
a unique $r \in \R^V$ that satisfies all the constraints.

First, since $\frac{df}{dt} = \Wm \frac{d \vp}{dt}$,
we have for each node $u \in V$,
$r_u = \frac{\rho_u}{w_u}$.

Rule \textsf{(R1)} implies the following constraint:

for $u \in V$ and $e \in E$,
$\rho_u(e) < 0$ only if $u \in S_e(f)$,
and $\rho_u(e) > 0$ only if $u \in I_e(f)$.

Rule \textsf{(R2)} implies the following constraint:

for each $e \in E$, we have
$\sum_{u \in I_e(f)} \rho_u(e) = - \sum_{u \in S_e(f)} \rho_u(e) = w_e \cdot \Delta_e(f)$.

Observe that for each $e \in E$,
once all the $\rho_u(e)$'s are determined,
the weight $w_e$ can be distributed among
edges in $S_e \times I_e$ by considering
a simple flow problem on the complete bipartite graph,
where each $u \in S_e$ is a source with
supply $-\frac{\rho_u(e)}{\Delta_e}$,
and each $v \in I_e$ is a sink with demand
$\frac{\rho_v(e)}{\Delta_e}$.  Then, from any feasible flow,
we can set $a^e_{uv}$ to be the flow along the edge $(u,v) \in S_e \times I_e$.

\noindent \textbf{Infinitesimal Considerations.}
In the previous discussion, we argue that
if a node $u$ is losing measure due to edge $e$,
then it should remain in $S_e$ for infinitesimal time,
which holds only if the rate of change of $f_u$ is the maximum among nodes in $S_e$.
A similar condition should hold if the node
$u$ is gaining measure due to edge $e$.  This translates
to the following constraints.

\noindent Rule \textsf{(R3)} First-Order Derivative Constraints:
\begin{compactitem}
\item If $\rho_u(e) < 0$, then $r_u \geq r_v$ for all $v \in S_e$.
\item If $\rho_u(e) > 0$, then $r_u \leq r_v$ for all $v \in I_e$.
\end{compactitem}

We remark that rule \textsf{(R3)}
is only a necessary condition in order
for the diffusion process to satisfy
rule~\textsf{(R1)}.  Even though $A_f$ might not be unique,
the following lemma shows that these rules
are sufficient to define a unique $r \in \R^V$.
Moreover, observe that if $f = \alpha g$ for some $\alpha > 0$, then
in the above flow problem to determine the symmetric matrix, we can still have $A_f = A_g$.
Hence, even though the resulting
$\Lo_w(f) := (\I - \Wm A_f) f$ might not be linear,
we still have \mbox{$\Lo_w (\alpha g) = \alpha \Lo_w(g)$}.

We also define
$r_S(e) := \max_{u \in S_e} r_u$ and $r_I(e) := \min_{u \in I_e} r_u$.

\begin{lemma}[Defining Laplacian from Diffusion Process]
\label{lemma:define_lap}
Given a measure vector $\vp \in \R^V$ (and $f = \Wm \vp$),
rules~\textsf{(R1)} to \textsf{(R3)}
uniquely determine $r \in \R^V$ (and $\rho = \W r$),
from which the operators $\Lo_w f := -r$
and \mbox{$\Lo \vp := - \W r$} are defined.
The normalized Laplacian is also induced
$\Lc := \Wmh \Lo \Wh$.

Furthermore, there is a diffusion process
that satisfies rules~\textsf{(R1)} to \textsf{(R3)},
and
can be described by the
differential equation: $\frac{df}{dt} = - \Lo_w f$.

Moreover,
$\sum_{e \in E} c_e (r_I(e) - r_S(e)) = \sum_{u \in V} \rho_u r_u = \|r\|^2_w$.
\end{lemma}

\begin{proof}
Observe that if the nodes $u$ all have distinct values for $f_u$'s, then the problem is trivial.  Define an equivalence relation on $V$ such that $u$ and $v$ are in the same equivalence class \emph{iff} $f_u = f_v$.  For each such equivalence
class $U \subset V$,
define $I_U := \{e \in E: \exists u \in U, u \in I_e\}$
and $S_U := \{e \in E: \exists u \in U, u \in S_e\}$.
Notice that each $e$ is in exactly one such $I$'s and one
such $S$'s.

As remarked above,
for each $e \in E$, once all $\rho_u(e)$ is defined
for all $u \in S_e \cup I_e$,
it is simple to determine $a^e_{uv}$ for $(u,v) \in S_e \times I_e$ by considering a flow problem on the bipartite graph $S_e \times I_e$.

\noindent \textbf{Considering Each Equivalence Class $U$.}
We can consider each equivalence class $U$ independently by analyzing $r_u$ and
$\rho_u(e)$ for $u \in U$ and $e \in I_U \cup S_U$ that satisfy rules~\textsf{(R1)} to \textsf{(R3)}.

\noindent \textbf{Proof of Uniqueness.}
 Our idea is to show that if there is some $r$ that can
satisfy rules~\textsf{(R1)} to \textsf{(R3)},
then it must take a unique value.  We also give a (not necessarily efficient) procedure
to construct $r$ and the relevant $\rho$'s.

For each $e \in I_U \cup S_U$,
recall that $c_e := w_e \cdot \Delta_e(f)$,
which is the rate of flow due to $e$ into $U$ (if $e \in I_U$) or
out of $U$ (if $e \in S_U$).
For $F \subseteq I_U \cup S_U$, denote $c(F) := \sum_{e \in F} c_e$.

Suppose $T$ is the set of nodes that have
the maximum $r$ values within the equivalence class,
i.e., for all $u \in T$, $r_u = \max_{v \in U} r_v$.
Observe that to satisfy rule~\textsf{(R3)}, for $e \in I_U$,
there is positive rate $c_e$ of measure flow into $T$ due to $e$
\emph{iff} $I_e \subseteq T$; otherwise, the entire rate $c_e$ will flow
into $U \setminus T$.
On the other hand, for $e \in S_U$,
if $S_e \cap T \neq \emptyset$, then there is a rate $c_e$
of flow out of $T$ due to $e$; otherwise,
the rate $c_e$ flows out of $U \setminus T$.

Based on this observation,
we define for $X \subset U$,
$I_X := \{e \in I_U: I_e \subseteq X\}$
and $S_X := \{e \in S_U: S_e \cap X \neq \emptyset\}$.
Note that these definitions are consistent with $I_U$ and $S_U$.
We denote $C(X) := c(I_X) - c(S_X)$.

To detect which nodes in $U$ should have the largest $r$ values,
we define $\delta(X) := \frac{C(X)}{w(X)}$,
which, loosely speaking, is the average weighted (with respect to $\W$) measure rate going into nodes in $X$. 
Observe that if $r$ is feasible, then the definition of $T$
implies that for all $v \in T$, $r_v = \delta(T)$.

We first prove that there exists
a unique maximal set $P$ with maximum average weighted measure rate
$\delta_M := \max_{\emptyset \neq Z \subseteq U} \delta(Z)$.

\noindent \textbf{Claim.}  Suppose $X, Y \subseteq U$ such that
$\delta(X) = \delta(Y) = \delta_M$.
Then, $\delta(X \cup Y) = \delta_M$.

\noindent \emph{Proof.}  Observe that $c(I_{X \cup Y}) \geq c(I_X) + c(I_Y) - c(I_{X\cap Y})$, and $c(S_{X \cup Y}) = c(S_X) + c(S_Y) - c(S_{X\cap Y})$.

Hence, $\delta(X \cup Y) \geq \frac{\delta_M (w(X) + w(Y)) - C(X \cap Y) }{w(X) + w(Y) - w(X \cap Y)}$, recalling that $\delta_M = \frac{C(X)}{w(X)} = \frac{C(Y)}{w(Y)}$.

Since $\delta(X) = \delta \geq \frac{C(X \cap Y)}{w(X \cap Y)}$,
we have $C(X \cap Y) \leq \delta_M \cdot w(X \cap Y)$.
  Therefore, we have
$\delta(X \cup Y) \geq \frac{\delta_M (w(X) + w(Y)) - \delta_M \cdot w(X \cap Y)}{w(X) + w(Y) - w(X \cap Y)} = \delta_M$, as required. \qed

Define $P := \cup_{X \subseteq U: \delta(X) = \delta_M} X$.
In view of the above claim, $\delta(P) = \delta_M$.
We next show that if $r$ is a feasible solution,
then $T = P$.  Observe that for a feasible $r$,
there must be at least rate of $c(I_P)$ going into $P$,
and at most rate of $c(S_P)$ going out of $P$.
Hence, we have $\sum_{u \in P} w_u r_u \geq c(I_P) - c(S_P) =  w(P) \cdot \delta(P)$.
Therefore, there exists $u \in P$ such that
$\delta(P) \leq r_u \leq \delta(T)$, where
the last inequality holds because every node $v \in T$ has $r_v = \delta(T)$.
This implies that $\delta(T) = \delta_M$, $T \subseteq P$
and the maximum $r$ value is $\delta_M = \delta(T) = \delta(P)$.  Therefore,
the above inequality becomes 
$w(P) \cdot \delta_M \geq \sum_{u \in P} w_u r_u \geq w(P) \cdot \delta(P)$,
which means equality actually holds.  This implies that $T = P$.

\noindent \emph{Recursive Argument.} Hence, it follows that
if $r$ is feasible, then $T$ can be uniquely identified (as
the maximal set $T$ having maximum $\delta(T)$),
and we must have for all $v \in T$, $r_v = \delta(T)$.
Then, the uniqueness argument can be applied
recursively for the smaller instance with
$U' := U \setminus T$, $I_{U'} := I_U \setminus I_T$,
$S_{U'} := S_U \setminus S_T$.

\noindent \textbf{Proof of Existence.}  We show
that once $T$ is identified above (by computing $\delta(Z)$ for all non-empty $Z \subseteq U$).  It is possible to
assign for each $v \in T$ and edge $e$ where $v\in I_e \cup S_e$,
the values $\rho_v(e)$ such that $\delta_M = r_v = \sum_{e} \rho_v(e)$.

Consider an arbitrary configuration $\rho$
in which edge $e \in I_T$ supplies a rate of $c_e$ to nodes in $T$,
and each edge $e \in S_T$ demands a rate of $c_e$ from nodes in $T$.
Each node $v \in T$ is supposed to gather a net rate of $w_v \cdot \delta_M$,
where any deviation is known as the \emph{surplus} or \emph{deficit}.

Given configuration $\rho$, define a directed graph $G_\rho$ with nodes in $T$
such that there is an arc $(u,v)$ if non-zero measure rate can be transferred from $u$ to $v$.
This can happen in one of two ways:
(i) there exists $e \in I_T$ containing both $u$ and $v$ such that
$\rho_u(e) > 0$, or
(ii) there exists $e \in S_T$ containing both $u$ and $v$ such that
$\rho_v(e) < 0$.

Hence, if there is a directed path from a node $u$ with non-zero surplus
to a node $v$ with non-zero deficit, then the surplus at node $u$ (and the
deficit at node $v$) can be decreased.

We argue that a configuration $\rho$ with minimum surplus must have zero surplus.
(Observe that the minimum can be achieved because $\rho$ comes from a compact set.)
Otherwise, suppose there is at least one node with positive surplus,
and let $T'$ be all the nodes that are reachable from some node with positive surplus in the
directed graph $G_\rho$. Hence, it follows that
for all $e \notin I_{T'}$, for all $v \in T'$, $\rho_v(e) = 0$,
and for all $e \in S_{T'}$, for all $u \notin T'$, $\rho_u(e)=0$.
This means that the rate going into $T'$ is $c(I_{T'})$ and all comes from $I_{T'}$,
and the rate going out of $T'$ is $c(S_{T'})$.  Since no node
in $T'$ has a deficit and at least one has positive surplus,
it follows that $\delta(T') > \delta_M$, which is a contradiction.

After we have shown that a configuration $\rho$ with zero surplus exists,
it can be found by a standard flow problem, in which each $e \in I_{T}$
has supply $c_e$, each $v \in {T}$ has demand $w_v \cdot \delta_M$,
and each $e \in S_{T}$ has demand $c_e$.  Moreover, in the flow network,
there is a directed edge $(e,v)$ if $v \in I_e$ and $(v,e)$ if $v \in S_e$.
Suppose in a feasible solution, there is a flow
with magnitude $\theta$ along a directed edge.
If the flow is in the direction $(e,v)$, then
$\rho_v(e) = \theta$; otherwise, if it is in the direction $(v,e)$,
then $\rho_v(e) = - \theta$.

\noindent \emph{Recursive Application.}
To show that the feasibility argument
can be applied recursively to the smaller instance $U'$,
we need to prove that $\delta_M' := \max_{\emptyset \neq Q \subset U'} \delta'(Q) < \delta_M$,
where $\delta'$ is the analogous function defined on $(U', I_{U'}, S_{U'})$.

Suppose $\emptyset \neq Q \subseteq U'$.
Because $T$ is the maximal set with maximum $\delta(T) = \delta_M$,
it follows that $C(T \cup Q) < \delta_M \cdot w(T \cup Q)$.
Hence, we have
$\delta'(Q) = \frac{c(I_{T \cup Q} \setminus I_T) - c(S_{T \cup Q} \setminus S_T)}{w(Q)} = \frac{C(T \cup Q) - C(T)}{w(Q)} < \frac{\delta_M \cdot w(T \cup Q) - \delta_M \cdot w(T)}{w(Q)} = \delta_M$, as required.

\noindent \textbf{Diffusion process is well-defined.}  In order to show that the
diffusion process can be described by
the differential equation $\frac{df}{dt} = - \Lo_w f$,
it suffices to check that for any time $t$,
there exists some $\eps>0$ such that $\frac{df}{dt}$
is continuous in $[t, t + \eps)$.  Recall that
we consider right-hand derivative, and the above
uniqueness and existence proof determines 
$\frac{df}{dt}$ at some time $t$.  Moreover, even though
$\frac{df}{dt}$ might not be continuous, $f$~is always continuous as $t$ increases.

Observe that as long as the equivalence classes
induced by $f$ do not change, then each of them act as a super node, and hence $r = \frac{df}{dt}$
remains continuous.  Observe that nodes from different equivalence classes have different $f$ values.  Since
$f$ is continuous, there exists $\eps > 0$ such that
there is no merging of equivalence classes in time $[t, t + \eps)$.  However, it is possible that an equivalence
class $U$ can split at time~$t$.

The first case is that the equivalence class $U$ is peeled off layer by layer in the recursive manner described above, where the set $T$ is such a layer.  Hence, the various $T$'s from $U$
actually behave like separate equivalence sub-classes, and
each of them has their own $r$ value.

The second case is more subtle, when the nodes in $T$
actually have their $f$ values split at time $t$,
and do not stay in the same equivalence class
after infinitesimal time.
For instance, there could be
a proper subset $X \subsetneq T$ whose $r$ values might be
marginally larger than the rest after infinitesimal time.

The potential issue is that
if the nodes in $X$ go on their own, then
the nodes $X$ and also the nodes
in $T \setminus X$ might experience a sudden jump in their rate $r$.

Fortunately, this cannot happen, because
we must have $\delta_M = \delta(T) = \delta(X)$.
Hence, after nodes in $X$ go on their own,
they will take care of $c(I_X)$ and $c(S_X)$,
leaving behind $c(I_T \setminus I_X)$ and 
$c(S_T \setminus S_X)$ for the set $T \setminus X$.
Hence,
in the remaining instance,
we still have $\delta'(T \setminus X) =
\frac{c(I_T \setminus I_X) - c(S_T \setminus S_X)}{w(T \setminus X)}
=\frac{C(T) - C(X)}{w(T) - w(X)} = \delta_M$.
This shows that the $r$ value for the nodes in $X$ and also nodes in $T\setminus X$
cannot suddenly jump.

\noindent \textbf{Claim.} $\sum_{e \in E} c_e (r_I(e) - r_S(e)) = \sum_{u \in V} \rho_u r_u$.

Consider $T$ defined above with $\delta_M = \delta(T) = r_u$ for $u \in T$.

Observe that $\sum_{u \in T} \rho_u r_u = (c(I_T) - c(S_T)) \cdot \delta_M
= \sum_{e \in I_T} c_e \cdot r_I(e) - \sum_{e \in S_T} c_e \cdot r_S(e)$,
where the last equality is due to rule~\textsf{(R3)}.

Observe that every $u \in V$ will be in exactly one such $T$, and
every $e \in E$ will be accounted for exactly once in each of $I_T$ and $S_T$, ranging over all $T$'s.  Hence, summing over all $T$'s gives the result.
\end{proof}

\subsection{Spectral Properties of Laplacian}
\label{sec:eigen}

We next consider the spectral properties
of the normalized Laplacian $\Lc$ induced
by the diffusion process defined in Section~\ref{sec:disp}.

\begin{lemma}[First-Order Derivatives]
\label{lemma:deriv}
Consider the diffusion process satisfying rules~\textsf{(R1)}
to~\textsf{(R3)} on 
the measure space with $\vp \in \R^V$, which
corresponds to $f = \Wm \vp$ in the weighted space.
Suppose 
$\Lo_w$ is the induced operator on the weighted space such that
$\frac{d f}{d t} = - \Lo_w f$.
Then, we have the following derivatives.

\begin{compactitem}
\item[1.] $\frac{d \|f\|^2_w}{dt} = - 2 \langle f, \Lo_w f \rangle_w$.
\item[2.] $\frac{d \langle f, \Lo_w f \rangle_w}{dt} 
= - 2 \|\Lo_w f \|^2_w$.
\item[3.] Suppose $\ray_w(f)$ is the Rayleigh quotient
with respect to the operator $\Lo_w$ on the weighted space.
Then, for $f \neq 0$, $\frac{d \ray_w(f)}{dt} = -\frac{2}{\|f\|^4_w} \cdot
(\|f\|^2_w \cdot \|\Lo_w f\|^2_w - \langle f , \Lo_w f \rangle^2_w) \leq 0$,
by the Cauchy-Schwarz inequality
on the $\langle \cdot , \cdot \rangle_w$ inner product, where equality
holds \emph{iff} $\Lo_w f \in \spn(f)$.
\end{compactitem}

\end{lemma}

\begin{proof}
For the first statement,
$\frac{d \|f\|_w^2}{d t} = 2 \langle f, \frac{d f}{d t} \rangle_w
= - 2 \langle f, \Lo_w f \rangle_w$.

For the second statement,
recall from Lemma~\ref{lemma:ray_disc}
that $\langle f, \Lo_w f \rangle_w = \sum_{e \in E} w_e \max_{u,v \in e} (f_u - f_v)^2$.
Moreover, recall also that
$c_e = w_e \cdot  \max_{u,v \in e} (f_u - f_v)$.
Recall that $r = \frac{df}{dt}$,
$r_S(e) = \max_{u \in S_e} r_u$
and $r_I(e) = \min_{u \in I_e} r_u$.

Hence, by the envelop theorem, $\frac{d \langle f, \Lo_w f \rangle_w}{dt} =
2 \sum_{e \in E} c_e \cdot (r_S(e) - r_I(e))$.
From Lemma~\ref{lemma:define_lap},
this equals $- 2 \|r\|^2_w = - 2 \|\Lo_w f\|^2_w$.

Finally, for the third statement, we have

$\frac{d}{dt} \frac{\langle f, \Lo_w f \rangle_w}{\langle f, f \rangle_w} = \frac{1}{\|f\|^4_w} (\| f \|^2_w \cdot \frac{d \langle f, \Lo_w f \rangle_w}{ dt} -  \langle f, \Lo_w f \rangle_w \cdot \frac{d \|f\|^2_w}{dt})
= -\frac{2}{\|f\|^4_w} \cdot
(\|f\|^2_w \cdot \|\Lo_w f\|^2_w - \langle f , \Lo_w f \rangle^2_w)$,
where the last equality follows from the first two statements.
\end{proof}

We next prove some properties
of the normalized Laplacian $\Lc$ with respect to orthogonal
projection in the normalized space.

\begin{lemma}[Laplacian and Orthogonal Projection]
\label{lemma:lap_proj}
Suppose $\Lc$ is the normalized Laplacian 
defined in Lemma~\ref{lemma:define_lap}.
Moreover, denote $x_1 := \Wh \vec{1}$, and
let $\Pi$ denote the orthogonal projection
into the subspace that is orthogonal to $x_1$.
Then, for all $x$, we have the following:
\begin{compactitem}
\item[1.] $\Lc(x) \perp x_1$,
\item[2.] $\langle x, \Lc x \rangle = \langle \Pi x, \Lc \Pi x \rangle$.
\end{compactitem}
\end{lemma}

\begin{proof}
For the first statement, observe that since the diffusion process
is defined on a closed system, the total measure given by $\sum_{u \in V} \vp_u$ does not change.
Therefore, $0 = \langle \vec{1}, \frac{d \vp}{d t} \rangle = \langle \Wh \vec{1}, \frac{d x}{d t} \rangle$,
which implies that $\Lc x = - \frac{d x}{d t} \perp x_1$.

For the second statement,
observe that from Lemma~\ref{lemma:ray_disc},
we have:

$\langle x, \Lc x \rangle = \sum_{e \in E} w_e
\max_{u,v \in e} (\frac{x_u}{\sqrt{w_u}} - \frac{x_v}{\sqrt{x_v}})^2 
= \langle (x + \alpha x_1), \Lc (x + \alpha x_1) \rangle $,
where the last equality holds for all real numbers $\alpha$.
It suffices to observe that $\Pi x = x + \alpha x_1$, for some suitable real $\alpha$.
\end{proof}

\begin{proofof}{Theorem~\ref{th:hyper_lap}}
Suppose $\Lc$ is the normalized Laplacian
induced by the diffusion process in Lemma~\ref{lemma:define_lap}.
Let $\gamma_2 := \min_{\vec{0} \neq x \perp \Wh \vec{1}} \rayc(x)$
be attained by some minimizer $x_2$.
We use the isomorphism between the three spaces:
$\Wmh \vp = x = \Wh f$.

The third statement of Lemma~\ref{lemma:deriv}
can be formulated in terms of the normalized space,
which states that $\frac{d \rayc(x)}{d t} \leq 0$,
where equality holds \emph{iff} $\Lc x \in \spn(x)$.

We claim that $\frac{d \rayc(x_2)}{d t} = 0$.
Otherwise, suppose $\frac{d \rayc(x_2)}{d t} < 0$.
From Lemma~\ref{lemma:lap_proj},
we have $\frac{dx}{dt} = - \Lc x \perp \Wh \vec{1}$.
Hence, it follows that at this moment, the current normalized
vector is at position $x_2$, and is moving
towards the direction given by
$x' := \frac{d x}{dt}|_{x=x_2}$ such that
$x' \perp \Wh \vec{1}$, and $\frac{d \rayc(x)}{d t}|_{x=x_2} < 0$.
Therefore, for sufficiently small $\eps > 0$,
it follows that $x_2' := x_2 + \eps x'$ is a non-zero vector
that is perpendicular to $\Wh \vec{1}$
and $\rayc(x_2') < \rayc(x_2) = \gamma_2$, contradicting the definition of $x_2$.

Hence, it follows that $\frac{d \rayc(x_2)}{d t} = 0$,
which implies that $\Lc x_2 \in \spn(x_2)$.
Since $\gamma_2 = \rayc(x_2) = \frac{\langle x_2, \Lc x_2 \rangle}{\langle x_2,  x_2 \rangle}$,
it follows that $\Lc x_2 = \gamma_2 x_2$, as required.
\end{proofof}

\section{Stochastic Diffusion Process}
\label{sec:stochastic}

In Section~\ref{sec:laplacian},
we define a diffusion process in a closed system with respect
to a hypergraph according to the
equation $\frac{d \vp}{d t} = - \Lo \vp$,
where $\vp \in \R^V$ is the measure vector,
and $\Lo$ is the corresponding operator on the measure space.
In this section, we consider the stochastic diffusion process
in which on the top of the diffusion process,
each node is subject to independent Brownian noise.
We analyze the process using \Ito calculus,
and the reader can refer to the textbook by
{\O}ksendal~\cite{oksendalSDE} for relevant background.

\noindent \textbf{Randomness Model.}
We consider the standard multi-dimensional Wiener process
$\{B_t \in \R^V: t \geq 0\}$ with independent Brownian
motion on each coordinate.
Suppose the variance of the Brownian motion 
experienced by each node is proportional to its weight.
To be precise, there exists $\eta \geq 0$
such that for each node $u \in V$,
the Brownian noise introduced to $u$ till time $t$ is $\sqrt{\eta w_u} \cdot B_t(u)$,
whose variance is $\eta w_u t$.
It follows that the net amount of measure
added to the system till time $t$
is $\sum_{u \in V} \sqrt{\eta w_u} \cdot B_t(u)$,
which has normal distribution  $N(0, \eta t \cdot w(V))$.
Observe that the special case for $\eta = 0$ is just the
diffusion process in a closed system.

This random model induces an \Ito process on the measure space
given by the following stochastic differential equation:
$d \Phi_t = - \Lo \Phi_t \, dt + \sqrt{\eta} \cdot \Wh \, d B_t$,
 with some initial
measure $\Phi_0$

By the transformation into the normalized space
$x := \Wmh \vp$, we consider the corresponding
stochastic differential equation in the normalized space:
$d X_t = - \Lc X_t \, dt + \sqrt{\eta} \, d B_t$,
where $\Lc$ is the normalized Laplacian from Lemma~\ref{lemma:define_lap}.  Observe that the random noise
in the normalized space is spherically symmetric.

\noindent \textbf{Convergence Metric.}
Given a measure vector $\vp \in \R^V$,
denote $\vp^* := \frac{\langle \vec{1}, \vp \rangle}{w(V)} \cdot \W \vec{1}$, which is the measure vector obtained by distributing
the total measure $\sum_{u \in V} \vp_u = \langle \vec{1}, \vp \rangle$ among the nodes such that each node $u$ receives an amount
proportional to its weight $w_u$.

For the normalized vector $x = \Wmh \vp$,
observe that $x^* := \Wmh \vp^* = \frac{\langle \vec{1}, \vp \rangle}{w(V)} \cdot \Wh \vec{1}$
is the projection of $x$ into the 
subspace spanned by $x_1 := \Wh \vec{1}$.
We denote by $\Pi$ the orthogonal projection
operator into the subspace orthogonal to $x_1$.

Hence, to analyze how far the measure is from being stationary,
we consider the vector $\Phi_t - \Phi^*_t$,
whose $\ell_1$-norm is
$\|\Phi_t - \Phi^*_t\|_1 \leq \sqrt{w(V)} \cdot \| \Pi X_t \|_2$.
As random noise is constantly delivered to the system,
we cannot hope to argue that these random quantities approach zero as $t$ tends to infinity.  However,
we can show that these random variables are stochastically
dominated by distributions with bounded mean and variance
as $t$ tends to infinity.  The following lemma states that a larger value 
of $\gamma_2$ implies that the measure is closer to being stationary.

\begin{lemma}[Stochastic Dominance]
\label{lemma:stoch_dom}
Suppose $\gamma_2 = \min_{0 \neq x \perp x_1} \rayc(x)$.
Then, in the stochastic diffusion process described above,
for each $t \ge 0$,
the random variable $\|\Pi X_t \|_2$
is stochastically dominated by 
$\|\widehat{X}_t \|_2$,
where $\widehat{X}_t$
has distribution 
$e^{-\gamma_2 t} \Pi {X}_0 + \sqrt{\frac{\eta}{2 \gamma_2} \cdot(1 - e^{- 2 \gamma_2 t})} \cdot N(0,1)^V$,
 and $N(0,1)^V$ is the standard $n$-dimensional Guassian distribution with independent coordinates.
\end{lemma}

\begin{proof}
Consider the function $h : \R^V \ra \R$ given by
$h(x) := \| \Pi x \|_2^2 = \| x -  x^*\|_2^2$,
where $x^* := \frac{\langle x_1, x \rangle}{w(V)} \cdot x_1$ and $x_1 := \Wh \vec{1}$.
Then, one can check that
the gradient is $\nabla h(x) = 2 \, \Pi x$,
and the Hessian is 
$\nabla^2 h(x) = 2 (\I - \frac{1}{w(V)} \cdot \Wh J \Wh)$,
where $J$ is the matrix where every entry is 1.

Define the \Ito process $Y_t := h(X_t) = \langle \Pi X_t,  \Pi X_t \rangle$.
By the \Ito lemma,
we have 

$d Y_t = \langle \nabla h(X_t), d X_t \rangle
+ \frac{1}{2} (d X_t)^\T \nabla^2 h(X_t) \, (d X_t)$.

To simplify the above expression, we make the substitution
$d X_t = - \Lc X_t \, dt + \sqrt{\eta} \, d B_t$.
From Lemma~\ref{lemma:lap_proj},
we have for all $x$, $\Lc x \perp x_1$ and
$\langle x , \Lc x \rangle = \langle \Pi x, \Lc \Pi x \rangle$.

Moreover, the convention for the product of differentials is
$0 = dt \cdot dt = dt \cdot d B_t(u) 
=  d B_t(u) \cdot d B_t(v)$ for $u \neq v$,
and $d B_t(u)\cdot d B_t(u) = dt$.
Hence, only the diagonal entries of the Hessian are relevant.

We have
$d Y_t = - 2 \langle \Pi X_t, \Lc \Pi X_t \rangle \, dt
+ \eta \sum_{u \in V} (1 - \frac{w_u}{w(V)}) \, dt
+ 2 \sqrt{\eta} \cdot \langle \Pi X_t, dB_t \rangle$.
Observing that $\Pi X_t \perp x_1$, from the definition of
$\gamma_2$, we have 
$\langle \Pi X_t, \Lc \Pi X_t \rangle \geq 
\gamma_2 \cdot \langle \Pi X_t, \Pi X_t \rangle$.
Hence, we have the following inequality:
$d Y_t \leq - 2 \gamma_2 Y_t \, dt
+ \eta n \, dt
+ 2 \sqrt{\eta} \cdot \langle \Pi X_t, dB_t \rangle$.

We next define another \Ito process $\widehat{Y_t} := \langle \widehat{X}_t, \widehat{X}_t \rangle$ with
initial value $\widehat{X}_0 := \Pi X_0$ and
stochastic differential equation:
$d \widehat{Y}_t = - 2 \gamma_2 \widehat{Y}_t \, dt
+ \eta n \, dt
+ 2 \sqrt{\eta} \cdot \langle \widehat{X}_t, d\widehat{B}_t \rangle$.

We briefly explain why $Y_t$ is stochastically dominated by
$\widehat{Y}_t$ by using a simple coupling argument.
If $Y_t < \widehat{Y_t}$, then we can choose $d B_t$ and $d\widehat{B}_t$ to be independent.
If $Y_t = \widehat{Y_t}$, observe that
$\langle \Pi X_t, dB_t \rangle$ and 
$\langle \widehat{X}_t, d\widehat{B}_t \rangle$ have the same distribution, because both $d B_t$ and $d\widehat{B}_t$
are spherically symmetric.  Hence,
in this case, we can choose a coupling between $d B_t$ and $d\widehat{B}_t$ such that
$\langle \Pi X_t, dB_t \rangle = \langle \widehat{X}_t, d\widehat{B}_t \rangle$.

Using the \Ito lemma, one can verify that the
above stochastic differential equation can be derived from
the following equation involving $\widehat{X}_t$:
$d \widehat{X}_t = - \gamma_2 \widehat{X}_t \, dt 
+ \sqrt{\eta} \, d \widehat{B}_t$.

Because $d \widehat{B}_t$ has independent coordinates,
it follows that the equation can be solved independently
for each node $u$.  Again, using
the \Ito lemma,
one can verify that $d(e^{\gamma_2 t} X_t) = \sqrt{\eta} \cdot e^{\gamma_2 t}
\, d \widehat{B}_t$.
Therefore, we have the solution
$\widehat{X}_t = e^{-\gamma_2 t} \widehat{X}_0 + \sqrt{\eta} \cdot e^{- \gamma_2 t} \int_{0}^t e^{\gamma_2 s} \, d \widehat{B}_s$,
which has the same distribution 
as: 

$e^{-\gamma_2 t} \widehat{X}_0 + \sqrt{\frac{\eta}{2 \gamma_2} \cdot(1 - e^{- 2 \gamma_2 t})} \cdot N(0,1)^V$, as required.
\end{proof}

\begin{corollary}[Convergence and Laplacian]
In the stochastic diffusion process, as $t$ tends to infinity,
$\|\Phi_t - \Phi_t^*\|_1^2$ is stochastically dominated by $\frac{\eta \cdot w(V)}{2 \gamma_2} \cdot \chi^2(n)$,
where $\chi^2(n)$ is the
chi-squared distribution with $n$ degrees of freedom.
Hence, $\lim_{t \ra \infty} E[ \|\Phi_t - \Phi_t^*\|_1]
\leq \sqrt{\frac{\eta n \cdot w(V)}{2 \gamma_2}}$.
\end{corollary}

\noindent \textbf{Remark.}
Observe that the total measure introduced
into the system is $\sum_{u \in V} \sqrt{\eta w_u} \cdot B_t(u)$,
which has standard deviation $\sqrt{\eta t \cdot w(V)}$.
Hence, as $t$ increases, the ``error rate''
is at most $\sqrt{\frac{n}{2 \gamma_2 t}}$.

\begin{proof}
Observe that, as $t$ tends to infinity,
$\widehat{Y_t} = \| \widehat{X}_t \|_2^2$
converges to the distribution $\frac{\eta}{2 \gamma_2} \cdot \chi^2(n)$, where $\chi^2(n)$ is the
chi-squared distribution with $n$ degrees of freedom (having
mean $n$ and standard deviation $\sqrt{2n}$).

Finally,
observing that $\|\Phi_t - \Phi_t^*\|_1^2 \leq w(V) \cdot \|\Pi X_t\|_2^2$,
it follows that as $t$ tends to infinity,
$\|\Phi_t - \Phi_t^*\|_1^2$ is stochastically dominated by
the distribution $\frac{\eta \cdot w(V)}{2 \gamma_2} \cdot \chi^2(n)$,
which has mean $\frac{\eta n \cdot w(V)}{2 \gamma_2}$
and standard deviation $\frac{\eta \sqrt{n} \cdot w(V)}{\sqrt{2} \gamma_2}$.
\end{proof}

\begin{corollary}[Upper Bound for Mixing Time for $\eta = 0$]
\label{cor:mixing}
Consider the deterministic diffusion process with $\eta = 0$,
and some initial probability measure $\vp_0 \in \R_+^V$
such that $\langle \vec{1}, \vp_0 \rangle = 1$.
Denote $\vp^* := \frac{1}{w(V)} \cdot \W \vec{1}$,
and $\vp^*_{\min} := \min_{u \in V} \vp^*(u)$.
Then, for any $\delta > 0$ and $t \geq \frac{1}{\gamma_2} \log \frac{1}{\delta \sqrt{\vp^*_{\min}}}$, we have
$\|\Phi_t - \vp^*\|_1 \leq \delta$.
\end{corollary}

\noindent \textbf{Remark.}
By considering edge expansion,
it is proved in~\cite{Louis15stoc} that
in their version of the diffusion process,
there exists some initial distribution $\vp_0$
such that for some $t \leq O(\frac{1}{\gamma_2} \log \frac{1}{\delta})$,
$\|\Phi_t - \vp^*\|_1 > \delta$.

\begin{proof}
In the deterministic process with $\eta = 0$, stochastic
dominance becomes
$\|\Pi X_t \|_2 \leq e^{\gamma_2 t} \cdot \|\Pi X_0\|_2$.

Relating the norms, we have
$\|\Phi_t - \vp^*\|_1 \leq \sqrt{w(V)} \cdot \|\Pi X_t\|_2
\leq \sqrt{w(V)} \cdot e^{- \gamma_2 t} \cdot \|\Pi X_0\|_2$.

Observe that
$\|\Pi X_0\|_2^2
\leq \langle X_0, X_0 \rangle
= \langle \vp_0, \Wm \vp_0 \rangle \leq \frac{1}{\min_u w_u}$.

Hence, it follows that
$\|\Phi_t - \vp^*\|_1 \leq \frac{1}{\sqrt{\vp^*_{\min}}} \cdot e^{- \gamma_2 t}$, which is at most $\delta$,
for $t \geq \frac{1}{\gamma_2} \log \frac{1}{\delta \sqrt{\vp^*_{\min}}}$.
\end{proof}

{
\bibliography{ref,cheeger,related_work}
\bibliographystyle{alpha}
}

\appendix

\section{Comparing Discrepancy Minimizers}
\label{sec:min}

We prove Lemma~\ref{lemma:min} by the following claims.

\begin{claim}
\label{claim:xg}
For $k \geq 1$, $\xi_k \leq \gamma_k$.
\end{claim}

\begin{proof}
Suppose the procedure produces $\{\gamma_i: i \in [k]\}$,
which is attained by orthonormal vectors $X_k := \{x_i : i \in [k]\}$.
Observe that $\max_{i \in [k]} \Dc(x_i) = \Dc(x_k) = \gamma_k$,
since $x_k$ could have been a candidate in the minimum for defining $\gamma_i$ because $x_k \perp x_j$, for all $j \in [k-1]$.

Since $X_k$ is a candidate for taking the minimum
over sets of $k$ orthonormal vectors in the definition of
$\xi_k$, it follows that $\xi_k \leq \gamma_k$.
\end{proof}

\begin{claim}
\label{claim:gz1}
For $k \geq 1$, $\gamma_k \leq \zeta_k$.
\end{claim}

\begin{proof}
For $k=1$, $\gamma_1 = \zeta_1 = 0$.

For $k > 1$, suppose the $\{\gamma_i : i \in [k-1]\}$
have already been constructed with
the corresponding orthonormal minimizers $X_{k-1} := \{x_i : i \in [k-1]\}$.

Let $Y_k := \{y_i: i \in [k]\}$ be an arbitrary set of $k$ orthonormal
vectors.  Since the subspace orthogonal to $X_{k-1}$
has rank $n - k + 1$ and the span of $Y_k$
has rank $k$, there must be a non-zero $y \in \spn(Y_k) \cap X_{k-1}^\perp$.

Hence, it follows that $\gamma_k = \min_{\vec{0} \neq x \in X_{k-1}^\perp}
\Dc(x) \leq \max_{y \in \spn(Y_k)} \Dc(y)$.
Since this holds for any set $Y_k$ of $k$ orthonormal vectors,
the result follows.
\end{proof}

\begin{claim}
\label{claim:zx}
For $k \geq 1$, $\zeta_k \leq k \xi_k$.
\end{claim}

\begin{proof}
Here it will be convenient to consider
the equivalent discrepancy ratios for the
weighted space.

Suppose $\xi_k$ is attained
by the orthonormal vectors $F_k := \{f_i : i \in [k]\}$ in
the weighted space, i.e., $\xi_k = \max_{i \in [k]} \D_w(f_i)$.
Then, it suffices to show that
for any $h \in \spn(F_k)$, $\D_w(h) \leq k \max_{i \in [k]} \D_w(f_i)$.

Suppose for some scalars $\alpha_i$'s, 
$
h=\sum_{i \in [k]} \alpha_i f_i.
$

For $u,v \in V$ we have
\begin{align*}
(h(u)-h(v))^2 
&= (\sum_{i\in[k]} \alpha_i(f_i(u)-f_i(v)))^2 \\
&\leq k\sum_{i\in [k]} \alpha_i^2 (f_i(u)-f_i(v))^2, \label{eq:1}
\end{align*}

where the last inequality follows from Cauchy-Schwarz inequality.
For each $e\in E$ we have 
\begin{align*}
\max_{u,v\in e} (h(u)-h(v))^2 
& \leq \max_{u,v\in e} k \sum_{i\in [k]} \alpha_i^2 (f_i(u)-f_i(v))^2 \\
& \leq k \sum_{i\in [k]} \alpha_i^2 \max_{u,v\in e} (f_i(u)-f_i(v))^2.
\end{align*} 

Therefore, we have
\begin{align*}
\D_w(h) &=  \frac{\sum_{e} \w{e} \max_{u,v\in e}{(h(u)-h(v))^2}}{\sum_{u \in V} \w{u} h(u)^2}\\
&\leq  \frac{k\sum_{i\in[k]} \alpha_i^2 \sum_{e} \w{e} \max_{u,v\in e}{(f_i(u)-f_i(v))^2}}{\sum_{i\in[k]} \alpha^2_i \sum_{u \in V} \w{u} f_i(u)^2} \\
&\leq k \max_{i \in [k]} \D_w(f_i),
\end{align*}
as required.
\end{proof}

\begin{claim}
\label{claim:gz2}
We have $\gamma_2 = \zeta_2$.
\end{claim}

\begin{proof}
From Claim~\ref{claim:gz1}, we already have $\gamma_2 \leq \zeta_2$.  Hence, it suffices to show the other direction.
We shall consider the discrepancy ratio for the weighted space.

Suppose $f \perp_w \mathbf{1}$ attains
$\D_w(f)=\gamma_2$. Then, we have
\begin{align*}
\zeta_2 &\leq \max_{g=af+b\mathbf{1}} \frac{\sum_{e \in E} \w{e} \max_{u,v\in e}{(\g{u}-\g{v})^2}}{\sum_{v \in V} \w{v} \g{v}^2} \\
&= \max_{g=af+b\mathbf{1}} \frac{\sum_{e \in E} \w{e} \max_{u,v\in e}{a^2(\f{u}-\f{v})^2}}{\sum_{v \in V} \w{v} (a \f{v}+b)^2}\\
&= \max_{g=af+b\mathbf{1}} \frac{\sum_{e \in E} \w{e} \max_{u,v\in e}{a^2(\f{u}-\f{v})^2}}{\sum_{v \in V} \w{v}(a^2 \f{v}^2+b^2)+2ab\sum_{v \in V} \w{v} \f{v}} \\
&\leq \max_{g=af+b\mathbf{1}} \frac{\sum_{e \in E} \w{e} \max_{u,v\in e}{a^2(\f{u}-\f{v})^2}}{\sum_{v \in V} a^2 \w{v} \f{v}^2}=\gamma_2.
\end{align*} 
\end{proof}

\section{Examples}
\label{sec:example}

We give examples of hypergraphs to show that
some properties are not satisfied.  For convenience,
we consider the properties in terms of the weighted space.
We remark that the examples could also be formulated
equivalently in the normalized space.
In our examples, the procedural minimizers are discovered by trial-and-error using programs.  However, we only
describe how to use Mathematica to 
verify them.  Our source code
can be downloaded at the following link:

\url{http://i.cs.hku.hk/~algth/project/hyper_lap/main.html}

\noindent \textbf{Verifying Procedural Minimizers.}
In our examples, we need to verify that
we have the correct value for
$\gamma_k := \min_{\vec{0} \neq f \perp_w \{f_1, f_2, \ldots, f_{k-1}\}} \D_w(f)$,
and a certain non-zero vector $f_k$ attains the minimum.

We first check that $f_k$ is perpendicular
to $\{f_1, \ldots, f_{k-1}\}$ in the weighted space,
and $\D_w(f_k)$ equals $\gamma_k$.

Then, it suffices to check that
for all $\vec{0} \neq f \perp_w \{f_1, f_2, \ldots, f_{k-1}\}$,
$\D_w(f) \geq \gamma_k$.  As the numerator in  the definition of $\D_w(f)$ involves
the maximum operator, we use a program to consider all cases of the relative
order of the nodes with respect to $f$.

For each permutation $\sigma: [n] \ra V$,
for $e \in E$, we define $S_\sigma(e) := \sigma(\max \{i: \sigma(i) \in e\})$
and $I_\sigma(e) := \sigma(\min \{i: \sigma(i) \in e\})$.

We consider the mathematical program
$P(\sigma) := \min \sum_{e \in E} w_e  \cdot (f(S_\sigma(e)) - f(I_\sigma(e)))^2
- \gamma_k \cdot \sum_{u \in V} w_u f(u)^2$ subject to
$f(\sigma(n)) \geq f(\sigma(n-1)) \geq \cdots f(\sigma(1))$
and $\forall i \in [k-1], \langle f_i, f \rangle = 0$.
Since the objective function is a polynomial, and all constraints
are linear, the Mathematica function \textsf{Minimize} can solve the
program.

Moreover, the following two statements are equivalent.
\begin{compactitem}
\item[1.] For all $\vec{0} \neq f \perp_w \{f_1, f_2, \ldots, f_{k-1}\}$,
$\D_w(f) \geq \gamma_k$.
\item[2.] For all permutations $\sigma$, $P(\sigma) \geq 0$.
\end{compactitem}

Hence, to verify the first statement, 
it suffices to use Mathematica to solve $P(\sigma)$ for all permutations $\sigma$.

\begin{example}
\label{ex:gamma_nonunique}
The sequence $\{\gamma_k\}$ generated
by the procedural minimizers is not unique.
\end{example}

\begin{proof}
Consider the following hypergraph with $5$ nodes and $5$ hyperedges each with unit weight.
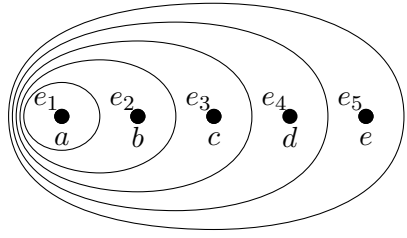
\begin{figure}[H]
	\centering
	\begin{tikzpicture}
	\node (v5) at (5,1) {};
	\node (v1) at (4,1) {};
	\node (v2) at (3,1) {};
	\node (v3) at (2,1) {};
	\node (v4) at (1,1) {};

	\begin{scope}[fill opacity=0]
	\filldraw[fill=blue!70] ($(v4)+(-0.5,0)$)
	to[out=90,in=180] ($(v4)+(0,0.45)$)
	to[out=0,in=90] ($(v4)+(0.5,0)$)
	to[out=270,in=0] ($(v4)+(0,-0.45)$)
	to[out=180,in=270] ($(v4)+(-0.5,0)$);
	\filldraw[fill=blue!70] ($(v4)+(-0.55,0)$)
	to[out=90,in=180] ($(v4)+(0.5,0.75)$)
	to[out=0,in=90] ($(v4)+(1.5,0)$)
	to[out=270,in=0] ($(v4)+(0.5,-0.75)$)
	to[out=180,in=270] ($(v4)+(-0.55,0)$);
	\filldraw[fill=blue!70] ($(v4)+(-0.6,0)$)
	to[out=90,in=180] ($(v4)+(1,1)$)
	to[out=0,in=90] ($(v4)+(2.5,0)$)
	to[out=270,in=0] ($(v4)+(1,-1)$)
	to[out=180,in=270] ($(v4)+(-0.6,0)$);
	\filldraw[fill=blue!70] ($(v4)+(-0.65,0)$)
	to[out=90,in=180] ($(v4)+(1.5,1.25)$)
	to[out=0,in=90] ($(v4)+(3.5,0)$)
	to[out=270,in=0] ($(v4)+(1.5,-1.25)$)
	to[out=180,in=270] ($(v4)+(-0.65,0)$);
	\filldraw[fill=blue!70] ($(v4)+(-0.7,0)$)
	to[out=90,in=180] ($(v4)+(2,1.5)$)
	to[out=0,in=90] ($(v4)+(4.5,0)$)
	to[out=270,in=0] ($(v4)+(2,-1.5)$)
	to[out=180,in=270] ($(v4)+(-0.7,0)$);
	\end{scope}

	\foreach \v in {1,2,...,5} {
		\fill (v\v) circle (0.1);
	}
	
	\fill (v1) node [below] {$d$};
	\fill (v2) circle (0.1) node[inner sep = 6pt] [below] {$c$};
	\fill (v3) circle (0.1) node [below] {$b$};
	\fill (v4) circle (0.1) node[inner sep = 6pt] [below] {$a$};
	\fill (v5) circle (0.1) node[inner sep = 6pt] [below] {$e$};
	
	\node at ($(v4)+(-0.2,0.2)$) {$e_1$};
	\node at ($(v3)+(-0.2,0.2)$) {$e_2$};
	\node at ($(v2)+(-0.2,0.2)$) {$e_3$};
	\node at ($(v1)+(-0.2,0.2)$) {$e_4$};
	\node at ($(v5)+(-0.2,0.2)$) {$e_5$};
	\end{tikzpicture}	
\end{figure}

We have verified that different minimizers for $\gamma_2$
can lead to different values for $\gamma_3$.

\begin{table}[H]
	\centering
	\begin{tabular}{c|cc|cc}
		\toprule
		{$i$} & $\gamma_i$ & $f_i^\T$ & $\gamma_i'$ & $f_i'^\T$\\
		\hline
		1 & 0 & $(1,1,1,1,1)$ & 0 & $(1,1,1,1,1)$ \\
		2 & 5/6 & $(1,1,1,-4,-4)$ & 5/6 & $(2,2,-3,-3,-3)$ \\
		3 & 113/99 & $(2,2,-6,3,-6)$ & 181/165 & $(4,-5,-5,5,5)$ \\
		\bottomrule
	\end{tabular}
	
\end{table}
\end{proof}

\begin{example}
\label{ex:xg}
There exists a hypergraph such that $\xi_2 < \gamma_2$.
\end{example}

\begin{proof}
Consider the following hypergraph $H = (V,E)$ with
$V = \{a, b, c, d\}$ and $E = \{e_i : i \in [5]\}$.
For $i \neq 3$, edge $e_i$ has weight 1, and edge $e_3$ has weight 2.
Observe that every node has weight $3$.

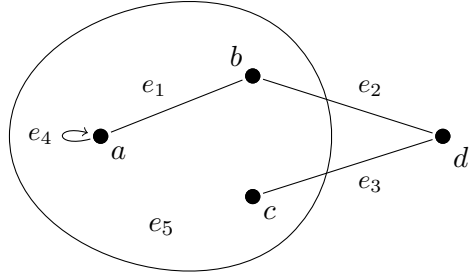
\begin{figure}[H]
	\centering
	\begin{tikzpicture}
	\node (v1) at (0,2) {};
	\node (v2) at (2,2.8) {};
	\node (v3) at (2,1.2) {};
	\node (v4) at (4.5,2) {};
	
	\begin{scope}[fill opacity=0]
	\filldraw ($(v2)+(+0.5,+0.5)$) 
	to[out=-45,in=45] ($(v3)+(0.5,-0.5)$) 
	to[out=225,in=270] ($(v1)+(-1.2,0)$)
	to[out=90,in=135] ($(v2)+(+0.5,+0.5)$);
	\end{scope}
	
	\path[every node/.style={font=\sffamily\small}]
	(v1) edge node [above left] {$e_1$} (v2)
	(v2) edge node [above right] {$e_2$} (v4)
	(v3) edge node [below right] {$e_3$} (v4)
	(v1) edge [loop left] node  {$e_4$} (v1);
	
	\foreach \v in {1,2,3,4} {
		\fill (v\v) circle (0.1);
	}
	
	\fill (v1) circle (0.1) node [below right] {$a$};
	\fill (v2) circle (0.1) node [above left] {$b$};
	\fill (v3) circle (0.1) node [below right] {$c$};
	\fill (v4) circle (0.1) node [below right] {$d$};
	
	\node at (0.8,0.8) {$e_5$};
	\end{tikzpicture}
\end{figure}

We can verify 
that $\gamma_2 = \frac{2}{3}$
with the corresponding vector $f_2 := (1, 1, -1, -1)^\T$.

Recall that $\xi_2=\min_{g_1,g_2} \max_{i\in [2]} \D_w(g_i)$,
where the minimum is over all non-zero $g_1$ and $g_2$ such that
$g_1 \perp_w g_2$.
We can verify that
$\xi_2 \leq \frac{1}{3}$ by considering the
the two orthogonal vectors 
$g_1=(0,0,1,1)^{\T}$ and $g_2=(1,1,0,0)^{\T}$
in the weighted space.
\end{proof}

\begin{example}[Issues with Operator in~\cite{Louis15stoc}]
\label{eg:Louis}
Suppose $\overline{\Lo}_w$ is the operator on the weighted space
that is derived from the description in~\cite{Louis15stoc}.
Then, there exists a hypergraph
such that any minimizer $f_2$ attaining
$\gamma_2 := \min_{\vec{0} \neq f \perp_w \vec{1}} \D_w(f)$
is not an eigenvector of $\overline{\Lo}_w$ or
even $\Pi^w_{\vec{1}^{\perp_w}} \overline{\Lo}_w$.
\end{example}

\begin{proof}
We use the same hypergraph as in Example~\ref{ex:xg}.
Recall that
$\gamma_2 = \frac{2}{3}$
with the corresponding vector $f_2 := (1, 1, -1, -1)^\T$.

We next show that $f_2$ is the only minimizer, up to scalar multiplication,
attaining $\gamma_2$.

According to the definition,
$$\gamma_2 = \min_{(a,b,c,d)\perp_w 1} 
\frac{(a-b)^2+(b-d)^2+2(c-d)^2+\max_{x,y\in e_5}(x-y)^2 }{3(a^2+b^2+c^2+d^2)}.$$

Without loss of generality, we only need to consider the following three cases:

\begin{enumerate}
\item[1.]$a\geq b \geq c$: Then, by substituting $a = -b-c-d$,
\begin{align*}
& \frac{(a-b)^2+(b-d)^2+2(c-d)^2+ (a-c)^2}{3(a^2+b^2+c^2+d^2)} \geq \frac{2}{3} \\
\iff & (c-d)^2+2(b+c)^2 \geq 0,
\end{align*}
and the equality is attained only when $a=b=-c=-d$.
\item[2.]$a\geq c\geq b$: Then, by substituting $d = -a-b-c$,
\begin{align*}
&  \frac{(a-b)^2+(b-d)^2+2(c-d)^2+ (a-b)^2}{3(a^2+b^2+c^2+d^2)} \geq \frac{2}{3} \\
\iff & (a+2b+c)^2+8c^2+4(a-c)(c-b) \geq 0,
\end{align*}
and the equality cannot be attained.
\item[3.]$b\geq a\geq c$: Then, by substituting $d = -a-b-c$,
\begin{align*}
&  \frac{(a-b)^2+(b-d)^2+2(c-d)^2+ (b-c)^2}{3(a^2+b^2+c^2+d^2)} \geq \frac{2}{3} \\
\iff & 4(b+c)^2+2(a+c)^2+2(b-a)(a-c) \geq 0,
\end{align*}
and the equality is attained only when $a=b=-c=-d$.
\end{enumerate}

Therefore, all minimizers attaining $\gamma_2$ must be in $\spn(f_2)$.

We next showt that $f_2$ is not an eigenvector of $\Pi^w_{\vec{1}^{\perp_w}} \overline{\Lo}_w$.
Observe that only the hyperedge $e_5 = \{a, b, c\}$
involves more than 2 nodes.  In this case,
the weight of $e_5$ is distributed evenly between
$\{a, c\}$ and $\{b,c\}$.  All other edges keep their weights.
Hence, the resulting weighted adjacency matrix $\overline{A}$ and
$\I - \Wm \overline{A}$ are as follows:

$\overline{A} = \begin{pmatrix}
\frac{3}{2} & 1 & \frac{1}{2} & 0 \\
1 & \frac{1}{2} & \frac{1}{2} & 1 \\
\frac{1}{2} & \frac{1}{2} & 0 & 2 \\
0 & 1 & 2 & 0
\end{pmatrix}$
and 
$\I - \Wm \overline{A} = \begin{pmatrix*}[r]
\frac{1}{2} & -\frac{1}{3} & -\frac{1}{6} & 0 \\
-\frac{1}{3} & \frac{5}{6} & -\frac{1}{6} & -\frac{1}{3} \\
-\frac{1}{6} & -\frac{1}{6} & 1 & -\frac{2}{3} \\
0 & -\frac{1}{3} & -\frac{2}{3} & 1
\end{pmatrix*}$

Hence, $\overline{\Lo}_w f_2 = (\I - \Wm \overline{A}) f_2
= (\frac{1}{3}, 1, -\frac{2}{3}, - \frac{2}{3})^\T \notin \spn(f_2)$.
Moreover, $\Pi^w_{\vec{1}^{\perp_w}} \overline{\Lo}_w f_2 =
(\frac{1}{3}, 1, -\frac{2}{3}, -\frac{2}{3})^\T
\notin \spn(f_2)$.

In comparison, in our approach,
since $b$ is already connected to $d$ with edge $e_2$ of weight 1,
it follows that the weight of $e_5$ should all go to the pair $\{a,c\}$.  Hence, the resulting adjacency matrix is:

$A = \begin{pmatrix}
1 & 1 & 1 & 0 \\
1 & 1 & 0 & 1 \\
1 & 0 & 0 & 2 \\
0 & 1 & 2 & 0
\end{pmatrix}.$

One can verify that $\Lo_w f_2 = (\I - \Wm A)f_2 = \frac{2}{3} f_2$, as claimed in Theorem~\ref{th:hyper_lap}.
\end{proof}

\begin{example}[Third minimizer not eigenvector of Laplacian]
\label{eg:gamma3}
There exists a hypergraph such that for all procedural minimizers $\{(f_i, \gamma_i)\}_{i\in [3]}$
of $\D_w$, the vector~$f_3$ 
is not an eigenvector of $\Lo_w$ or even
$\Pi^w_{F_2^{\perp_w}} \Lo_w$,
where $\Lo_w$ is the operator on the weighted space
defined in Lemma~\ref{lemma:define_lap},
and $F_2 := \{f_1, f_2\}$.
\end{example}

\begin{proof}
Consider the following hypergraph with $4$ nodes and $2$ hyperedges each with unit weight.
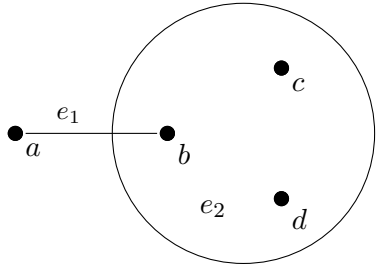
\begin{figure}[H]
	\centering
	\begin{tikzpicture}
	\node (v1) at (0,0) {};
	\node (v2) at (1.5,0.866) {};
	\node (v3) at (1.5,-0.866) {};
	\node (v4) at (-2,0) {};
	
	\draw (1,0) circle (1.725);	
	
	\path[every node/.style={font=\sffamily\small}]
		(v1) edge node [above left] {$e_1$} (v4);
	
	\foreach \v in {1,2,3,4} {
		\fill (v\v) circle (0.1);
	}
	
	\fill (v1) circle (0.1) node [below right] {$b$};
	\fill (v2) circle (0.1) node [below right] {$c$};
	\fill (v3) circle (0.1) node [below right] {$d$};
	\fill (v4) circle (0.1) node [below right] {$a$};
	
	\node at (0.6,-1) {$e_2$};
	\end{tikzpicture}
\end{figure}	

We can verify the first 3 procedural minimizers.

\begin{table}[H]
	\centering
	\begin{tabular}{l|cc}
		\toprule
		{$i$} & $\gamma_i$ & $f_i^\T$ \\
		\hline
		1 & 0 & $(1,1,1,1)$  \\
		2 & $\frac{5-\sqrt{5}}{4}$ & $(\sqrt{5}-1,\frac{3-\sqrt{5}}{2},-1,-1)$  \\
		3 & $\frac{11+\sqrt{5}}{8}$ & $(\sqrt{5}-1, -1, 4-\sqrt{5}, -1)$  \\
		3$'$& $\frac{11+\sqrt{5}}{8}$ & $(\sqrt{5}-1, -1, -1, 4-\sqrt{5})$ \\
		\bottomrule
	\end{tabular}
\end{table}

We next show that $f_3$ and $f_{3'}$ are the only minimizers, up to scalar multiplication,
attaining $\gamma_3$.

According to the definition,
$$\gamma_2 = \min_{(a,b,c,d)\perp_w 1}\frac{(a-b)^2+\max_{x,y \in e_2}(x-y)^2}{a^2+2b^2+c^2+d^2}.$$

Observe that $c$ and $d$ are symmetric, we only need to consider the following two cases,
\begin{enumerate}
	\item[1.]$c\geq b \geq d$: Then, by substituting $a=-2b-c-d$,
	\begin{align*}
		& \frac{(a-b)^2+(c-d)^2}{a^2+2b^2+c^2+d^2} \geq 1 \\
		\iff & 5b^2 + 2(c-b)(b-d) \geq 0.
	\end{align*}
	\item[2.]$b \geq c \geq d$: Then, by substituting $a=-2b-c-d$,
	\begin{align*}
		& \frac{(a-b)^2+(b-d)^2}{a^2+2b^2+c^2+d^2} \geq \frac{5-\sqrt{5}}{4} \\
		\iff & (5+3\sqrt{5})b^2+(\sqrt{5}-3)c^2+(\sqrt{5}-1)d^2 + (2\sqrt{5}+2)bc+(2\sqrt{5}-2)bd+(\sqrt{5}-1)cd \geq 0.
	\end{align*}
	Let $f(b,c,d)$ denotes the function above.
	Since $f$ is a quadratic function of $c$ and the coefficient of $c^2$ is negative, the minimum must be achieved when $c=b \text{ or } d$.
	In other words,$f(b,c,d)\geq \min \{f(b,b,d), f(b,d,d)\}.$ Note that
	\begin{align*}
		& f(b,b,d) = (6\sqrt{5}+4)b^2+(3\sqrt{5}-3)bd+(\sqrt{5}-1)d^2 \geq 0 \\
		\text{and } & f(b,d,d) = (5+3\sqrt{5})b^2 +4\sqrt{5}bd+(3\sqrt{5}-5)d^2 \geq 0.
	\end{align*}
	and the equality holds only when $c=d=-\frac{3+\sqrt{5}}{2}b$.
\end{enumerate}
Therefore, $\gamma_2 = \frac{5-\sqrt{5}}{4}, f_2^{T}=(\sqrt{5}-1,\frac{3-\sqrt{5}}{2},-1,-1)$.

Now we are ready to calculate $\gamma_3.$
$$\gamma_3 = \min_{(a,b,c,d)\perp_w 1,f_2}\frac{(a-b)^2+\max_{x,y \in e_2}(x-y)^2}{a^2+2b^2+c^2+d^2}.$$
Note that,
$$(a,b,c,d)\perp_w \vec{1},f_2 \iff
\begin{cases}
a+2d+c+d = 0 \\
(\sqrt{5}-1)a+(3-\sqrt{5})b-c-d=0
\end{cases} \iff
\begin{cases}
a = (1-\sqrt{5})b \\
c+d = (\sqrt{5}-3)b
\end{cases}
$$
\begin{enumerate}
	\item[1.]$c\geq b \geq d$: which is equivalent to $c \geq -\frac{\sqrt{5}+3}{4}(c+d) \geq d,$ then
	\begin{align*}
		& \frac{(a-b)^2+(c-d)^2}{a^2+2b^2+c^2+d^2} \geq \frac{11+\sqrt{5}}{8} \\
		\iff & (c-\frac{\sqrt{5}+3}{4}(c+d))(d-\frac{\sqrt{5}+3}{4}(c+d)) \leq 0.
	\end{align*} 
	\item[2.]$b\geq c \geq d$: which is equivalent to $(4-\sqrt{5})b+d \geq 0 \geq (3-\sqrt{5})b+2d,$ then
	\begin{align*}
		& \frac{(a-b)^2+(b-d)^2}{a^2+2b^2+c^2+d^2} \geq \frac{11+\sqrt{5}}{8} \\
		\iff & ((4-\sqrt{5})b + d)((3+\sqrt{5})((3-\sqrt{5})b+2d)- (\sqrt{5}-1)((4-\sqrt{5})b+d)) \leq 0.
	\end{align*}
\end{enumerate}
Therefore, $\gamma_3=\frac{11+\sqrt{5}}{8}$, and the corresponding $f_3^{T}=(\sqrt{5}-1, -1, 4-\sqrt{5}, -1) \text{ or } (\sqrt{5}-1, -1, -1, 4-\sqrt{5}).$

We let $f= f_3 = (\sqrt{5}-1, -1, 4-\sqrt{5}, -1)^\T$,
and we apply the procedure described in Lemma~\ref{lemma:define_lap} to compute $\Lo_w f$.

Observe that $w_a = w_c = w_d = 1$ and $w_b=2$,
and $f(b) = f(d) < f(a) < f(c)$.

For edge $e_1$, $\Delta_1 = f(a) - f(b) = \sqrt{5}$ and $c_1 = w_1 \cdot \Delta_1 = \sqrt{5}$.
For edge $e_2$, $\Delta_2 = f(c) - f(b) = 5 - \sqrt{5}$,
and $c_2 = w_2 \cdot \Delta_2 = 5 - \sqrt{5}$.
Hence, $r_a = - \frac{c_1}{w_a}$, $r_c = - \frac{c_2}{w_c}$,
and $r_b = r_d = \frac{c_1+c_2}{w_b + w_d}$.

Therefore, $\Lo_w f = - r = 
(\sqrt{5}, - \frac{5}{3}, 5 - \sqrt{5},
- \frac{5}{3})^\T$.

Moreover, $\Pi^w_{F_2^{\perp_w}} \Lo_w f =
(-\frac{1}{2} + \frac{7}{6} \cdot \sqrt{5},
-\frac{4}{3} - \frac{1}{6} \cdot \sqrt{5},
\frac{59}{12} - \frac{11}{12} \cdot \sqrt{5},
-\frac{7}{4} + \frac{1}{12} \cdot \sqrt{5})^\T
\notin \spn(f)$.

The case when $f_3 = (\sqrt{5}-1, -1, -1, 4-\sqrt{5})^\T$
is similar, with the roles of $c$ and $d$ reversed.
\end{proof}

\end{document}